\documentclass[journal,english]{IEEEtran}

\usepackage{cite}
\usepackage{mdframed}
\usepackage{epsfig,subfigure}

\usepackage{amssymb}
\usepackage{amsmath}
\usepackage{color}

\allowdisplaybreaks
\def\BibTeX{{\rm B\kern-.05em{\sc i\kern-.025em b}\kern-.08em
    T\kern-.1667em\lower.7ex\hbox{E}\kern-.125emX}}
\usepackage{balance}

\usepackage{hyperref}

\usepackage{arydshln}

\usepackage{graphicx}
\usepackage{color}
\usepackage[dvipsnames]{xcolor}

\definecolor{armygreen}{rgb}{0.29, 0.33, 0.13}

\usepackage{amssymb}

\usepackage{amsmath,amsfonts,amssymb}

\usepackage{verbatim}
\usepackage{stfloats}
\usepackage[bookmarks=false]{}
\usepackage{algorithm}

\usepackage{algcompatible}

\usepackage{tikz,pgfplots,tikzpeople}
\pgfplotsset{compat=1.18}

\newtheorem{theorem}{Theorem}

\newtheorem{lemma}{Lemma}

\newtheorem{remark}{Remark}

\begin{document}
\date{}

\title{On the Capacity of Hierarchical Secure Aggregation with 
Groupwise Keys}

\author{\IEEEauthorblockN{
Minyang Lu\IEEEauthorrefmark{1},
Zhou Li\IEEEauthorrefmark{1},
Haiqiang Chen\IEEEauthorrefmark{1},
and Min Xie\IEEEauthorrefmark{1} 
}\\
\IEEEauthorblockA{
Guangxi Key Laboratory of Multimedia Communications and Network Technology, Guangxi University\IEEEauthorrefmark{1}\\
Email:~\IEEEauthorrefmark{1}luminyang@st.gxu.edu.cn,~\IEEEauthorrefmark{1}\{lizhou,haiqiang\}@gxu.edu.cn,~\IEEEauthorrefmark{1}gxxiemin@163.com}
}

\IEEEpeerreviewmaketitle
\pgfplotsset{compat=1.18}

\maketitle
\begin{abstract}
We study the hierarchical secure aggregation problem with groupwise keys. The problem consists of an aggregation server, $U$ relays, and $UV$ users, where each relay serves $V$ disjoint users, and each subset of $G$ users shares an independent groupwise key. Two security requirements are imposed: relay security and server security. Specifically, each relay must not learn any information about the users' inputs, and the server must not learn any additional information beyond the recovered sum of all inputs.

We first show that the problem is infeasible when $G = 1$. For the feasible regime $1 < G \le UV$, we fully characterize the optimal rate region. In particular, we prove that both each user and each relay must transmit at least one symbol per input symbol. Furthermore, we characterize the minimum required groupwise key rate as
$\max\left\{\frac{V}{\binom{UV}{G} - \binom{(U-1)V}{G}},\;
\frac{U - 1}{\binom{UV}{G} - U \binom{V}{G}}\right\},$
where the two terms correspond to the constraints imposed by relay security and server security, respectively.

For achievability, we propose an explicit linear coding scheme based on structured precoding matrices, and show that it satisfies both correctness and security requirements. The construction avoids permutation-based symmetrization by leveraging sufficiently generic matrix designs over large fields. Finally, we establish a matching converse, thereby characterizing the optimal rate region.
\end{abstract}


\begin{IEEEkeywords}
Hierarchical secure aggregation, groupwise keys, information theoretic security, groupwise key. 
\end{IEEEkeywords}

\section{Introduction}

Secure aggregation is a fundamental primitive in distributed computation~\cite{Blaum_Bruck_Vardy,Dimakis_survey,VJ_Survey} and federated learning~\cite{konecny2016federated,kairouz2021advances,rieke2020future}. In this setting, an aggregation server aims to compute a prescribed function of users’ private inputs while learning no additional information beyond the desired output. Among various functions, secure summation has received particular attention, as it serves as a core building block for privacy-preserving learning and large-scale data analytics.

From an information-theoretic perspective, secure summation has been extensively studied under a variety of network models and adversarial assumptions. Early works focused on the single-server setting and established fundamental limits on communication efficiency, key size, and robustness against user collusion~\cite{9834981,zhao2023secure}. Subsequent research has significantly broadened this line of work, incorporating practical constraints such as user dropout~\cite{so2022lightsecagg}, and optimal constructions based on groupwise or MDS-type key structures~\cite{zhao2023secure,wan2024information,wan2024capacity,Li_Zhang_GroupwiseDSA}. More recently, decentralized variants have also been investigated~\cite{Zhang_Li_Wan_DSA,Li_Zhang_GroupwiseDSA,Li_Zhang_WeaklyDSA}. A central insight emerging from these works is that groupwise keys, shared among subsets of users, enable optimal communication–key tradeoffs in flat network settings.

Motivated by large-scale deployments in federated learning, hierarchical architectures have recently attracted increasing attention. In such systems, users are organized into clusters, each connected to an intermediate relay, which in turn communicates with a central aggregation server. This architecture naturally arises due to scalability and communication constraints, but also introduces new security challenges, as both relays and the central server may act as honest-but-curious adversaries.

To capture these challenges, Zhang \emph{et al.}~\cite{zhang2024optimal} introduced the hierarchical secure aggregation (HSA) framework, providing an information-theoretic formulation for multi-layer aggregation systems. Follow-up works have further explored HSA under various models, including wireless settings, collusion constraints, and heterogeneous system assumptions~\cite{10806947,egger2024privateaggregationhierarchicalwireless,zhang2025fundamental,11195652,li2025collusionresilienthierarchicalsecureaggregation,egger2023private,lu2024capacity,Li_Zhang_WeaklyHSA}. These studies have established feasibility conditions and characterized communication or key requirements under specific constructions.

However, most existing HSA schemes rely on pairwise or cluster-based key designs that are tightly coupled to the hierarchical topology. While such constructions achieve desired security guarantees, they restrict the key-sharing structure and do not fully exploit the flexibility offered by groupwise keys. As a result, it remains unclear whether hierarchical secure aggregation can benefit from the same optimality properties of groupwise key designs observed in flat settings.

In this paper, we address this gap by studying hierarchical secure aggregation with groupwise keys. We consider a three-layer network consisting of an aggregation server, $U$ relays, and $UV$ users, where each relay serves a disjoint cluster of $V$ users (see Fig.~\ref{fig:my_image}). Each subset of $G$ users shares an independent groupwise key, and each user holds all keys corresponding to groups it belongs to. We impose two security requirements: relay security and server security. The goal is to securely compute the sum of all users’ inputs while characterizing the fundamental tradeoffs between communication and key resources.

Our main contribution is a complete information-theoretic characterization of this problem. We first establish feasibility conditions, showing that secure aggregation is impossible when $G=1$. For the feasible regime $1 < G \le UV$, we prove that each user and each relay must transmit at least one symbol per input symbol, establishing fundamental communication lower bounds. We further characterize the minimum required groupwise key rate and show that it decomposes into two components corresponding to relay-side and server-side security constraints, with the overall requirement determined by the more stringent one.

To establish achievability, we construct explicit linear schemes based on groupwise keys and structured precoding matrices. These schemes ensure correct aggregation while guaranteeing information-theoretic security at both the relay side and the server side. Several examples are provided to illustrate the constructions and to highlight how the required key resources scale with system parameters.

Overall, our results extend the theory of secure aggregation from flat networks to hierarchical architectures with general groupwise key structures, providing new insights into the interplay between communication efficiency, key design, and security in distributed systems.

\begin{figure}[!hbt]
    \centering
\includegraphics[width=0.45\textwidth,height=0.35\textwidth, keepaspectratio]{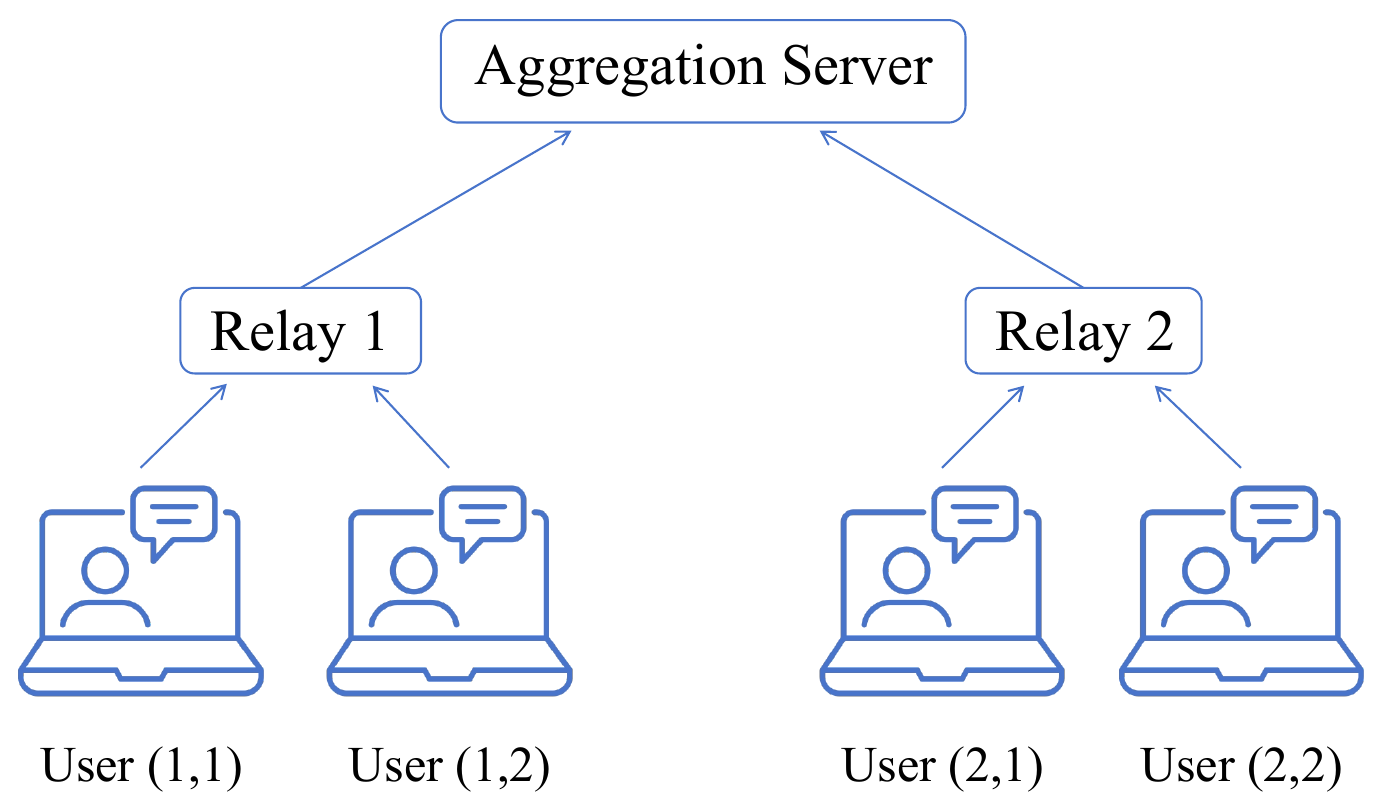}
 \caption{\small The hierarchical secure aggregation method with $(U, V,G) = (2, 2 ,2)$. In this figure, there is a aggregation server and two relays. Each relay is connected to two users. $G = 2$ means that every 2 users share an equal-sized independent key.} 
    \label{fig:my_image} 
\end{figure}
\bigskip
{\em Notation: Bold capital letters, bold lowercase letters, and calligraphic fonts will be used to denote arrays, vectors, and sets, respectively. For positive integers $K_1$, $K_2$, $K_1 < K_2$, we use the notation $[K_1:K_2] \triangleq \{K_1, K_1+1,\cdots, K_2\}$, and $[1:K_2]$ is abbreviated as $[K_2]$. 
Denote $\{A_i\}_{i\in[n]} \stackrel{\triangle}{=} \{A_1, \dots, A_n\}$. $a \bmod b$ represents the modulo operation on $a$ with integer divisor $b$ and in this paper we let $(a \bmod b) \in [0:b]$ (i.e., we let  $a \bmod b$ = $b$ if $b$ divides $a$), define that $\langle b \rangle_a = b \bmod a$. The notation $|\mathcal{A}|$ is used to denote the cardinality of a set $\mathcal{A}$. The notation $|\mathcal{B}|$ is used to denote the cardinality of a set $\mathcal{B}$. For two sets $\mathcal{A}$ and $\mathcal{B}$, we use $\mathcal{A}\backslash\mathcal{B}$ to denote the set of elements that belong to $\mathcal{A}$ but not $\mathcal{B}$. The notation $\binom{\mathcal{A}}{G}$ is used to denote all subsets of $\mathcal{A}$ with cardinality $G$, i.e., $\binom{\mathcal{A}}{G} \triangleq \{\mathcal{G}:\mathcal{G}\subset{\mathcal{A}},|\mathcal{G}|=G\}$, if $|\mathcal{A}| < G$, then $\binom{|\mathcal{A}|}{G} = 0$. Denote $\mathcal{B} = \{B_1,\cdots,B_{|\mathcal{B}|}\}$, denote $\mathcal{B}\times [K_2] \stackrel{\triangle}{=} \{(B_1,1),\cdots,(B_1,K_2),\cdots,(B_{|\mathcal{B}|},1),\cdots,(B_{|\mathcal{B}|},K_2)\}$. Let $a$ and $b$ be positive integers. Define $c = \left\lceil \frac{b}{a} \right\rceil$ as the ceiling of $b/a$, i.e., the smallest integer satisfying $c \ge b/a$.}

\section{Problem Statement}

We consider a secure aggregation problem over a three-layer hierarchical network consisting of one central server, $U \geq 2$ relays, and $UV$ users. Each relay is connected to a disjoint group of $V$ users. Communication occurs over two uplink hops: from users to their associated relays, and from relays to the central server (see Fig.~\ref{fig:my_image} for an example). All links are assumed to be noiseless.
The $v$-th user connected to Relay $u$ is indexed by $(u,v) \in [U] \times [V]$. The set of users associated with Relay $u$ is defined as
$\mathcal{M}_u \triangleq \{(u,v)\}_{v \in [V]}.$
Each user $(u,v)$ holds an input vector $W_{u,v} \in \mathbb{F}_q^L$. Let
$W_{[U]\times[V]} \triangleq \{W_{u,v}\}_{(u,v)\in[U]\times[V]}$
denote the collection of all user inputs. The inputs are assumed to be mutually independent and uniformly distributed, i.e.,
\begin{align}
H\big(W_{[U]\times[V]}\big)
&= \sum_{(u,v)\in[U]\times[V]} H(W_{u,v}), \label{independent} \\
H(W_{u,v}) &= L, \quad \forall (u,v)\in[U]\times[V]. \label{inputsize}
\end{align}

Each user $(u,v)$ is equipped with a private key $Z_{u,v}$ with entropy $H(Z_{u,v}) = L_Z$. Let
$Z_{[U]\times[V]} \triangleq \{Z_{u,v}\}_{(u,v)\in[U]\times[V]}$
denote the collection of all user keys. 
We assume the existence of a trusted third party that generates and distributes the keys.
Furthermore, the inputs and keys are independent, i.e.,
\begin{align}
&H\big(W_{[U]\times[V]}, Z_{[U]\times[V]}\big)\notag\\
=& \sum_{(u,v)\in[U]\times[V]} H(W_{u,v})
+ H\big(Z_{[U]\times[V]}\big). \label{independent2}
\end{align}

In this work, we adopt a symmetric \emph{groupwise key} structure. Each key is shared among a group of $G \in [UV]$ users, and all such keys are mutually independent.
Specifically, let $\binom{[U]\times[V]}{G}$ denote the collection of all user subsets of size $G$. For each group $\mathcal{G} \in \binom{[U]\times[V]}{G}$, we generate an independent random variable $S_{\mathcal{G}}$, consisting of $L_S$ i.i.d. symbols over $\mathbb{F}_q$. The key $S_{\mathcal{G}}$ is shared among all users in $\mathcal{G}$.
The key variable available at user $(u,v)$ is given by
\begin{equation}
Z_{u,v}
= \{ S_{\mathcal{G}} \}_{\mathcal{G} \in \binom{[U]\times[V]}{G},\; (u,v)\in\mathcal{G}}, \quad \forall (u,v)\in[U]\times[V].
\label{individualkey}
\end{equation}

All keys $\{S_{\mathcal{G}}\}$ are mutually independent and independent of $\{W_{u,v}\}$. Therefore, we have
\begin{align}
&H\big( W_{[U]\times[V]}, Z_{[U]\times[V]} \big) \notag\\
=& H\big( W_{[U]\times[V]}, \{S_{\mathcal{G}}\}_{\mathcal{G} \in \binom{[U]\times[V]}{G}} \big) \notag\\
=& \sum_{(u,v)\in[U]\times[V]} H(W_{u,v})
+ \sum_{\mathcal{G} \in \binom{[U]\times[V]}{G}} H(S_{\mathcal{G}}) \notag\\
=& \sum_{(u,v)\in[U]\times[V]} H(W_{u,v}) + \binom{UV}{G} L_S.
\label{gruopwisesize}
\end{align}

The hierarchical secure aggregation scheme operates over two hops.

\textbf{First Hop:}
Each user $(u,v) \in [U]\times[V]$ transmits a message $X_{u,v}$ of entropy $H(X_{u,v}) = L_X$ to Relay $u$, which is generated as a function of $W_{u,v}$ and $Z_{u,v}$. This implies
\begin{equation}
H(X_{u,v} \mid W_{u,v}, Z_{u,v}) = 0, \quad \forall (u,v)\in[U]\times[V].
\label{messageX}
\end{equation}

\textbf{Second Hop:}
Each relay $u \in [U]$ transmits a message $Y_u$ of entropy $H(Y_u) = L_Y$ to the aggregation server, which is generated as a function of the received messages $\{X_{u,v}\}_{v\in[V]}$. This implies
\begin{equation}
H\big(Y_u \mid \{X_{u,v}\}_{v\in[V]}\big) = 0, \quad \forall u\in[U].
\label{messageY}
\end{equation}

Based on the above two-hop communication structure, the aggregation server is required to reliably compute the desired aggregation of the user inputs. In this work, the goal is to recover the sum of all inputs.

\textit{Correctness:}
The server is required to reliably compute the sum of all user inputs $\sum_{(u,v)\in[U]\times[V]} W_{u,v}$ from the relay transmissions $\{Y_u\}_{u\in[U]}$. In particular, perfect recovery is required, which is captured by
\begin{equation}
H\Bigg( \sum_{(u,v)\in [U]\times [V]} W_{u,v} \,\Big|\, \{Y_u\}_{u \in [U]} \Bigg) = 0.
\label{Correctness}
\end{equation}

We now formalize the security requirements of the problem. In addition to correctness, we impose security guarantees against both intermediate relays and the final aggregation server.

\textit{Relay Security:} It requires that even if relay $u$ receives the information $\{X_{u,v}\}_{v \in [V]}$ from the connected users, it cannot infer any information about the input $W_{u,v}$. That is to say, 
\begin{equation}
I\big( \{X_{u,v}\}_{v\in[V]}; W_{[U]\times[V]}  \big) = 0,
\quad \forall u\in[U].
\label{RelaySecurity}
\end{equation}

\textit{Server Security:}
The aggregation server, upon receiving the messages $\{Y_u\}_{u\in[U]}$ from the relays, cannot obtain any information about the individual inputs $\{W_{u,v}\}_{(u,v)\in[U]\times[V]}$ beyond the intended total sum $\sum_{(u,v)\in[U]\times[V]} W_{u,v}$. Formally, this is characterized by the conditional mutual information

\begin{equation}
I\Big( \{Y_u\}_{u\in[U]}; W_{[U]\times[V]} \,\Big|\, \sum_{(u,v)\in[U]\times[V]} W_{u,v} \Big) = 0.
\label{ServerSecurity}
\end{equation}

Based on the above problem formulation, we define the communication and key rates of the proposed hierarchical secure aggregation scheme.

The communication rates $R_X$ and $R_Y$ represent the number of $q$-ary symbols transmitted in each message $X_{u,v}$ and $Y_u$ per source symbol, respectively, while $R_S$ denotes the required groupwise key rate per source symbol. The rates are defined as
\begin{equation}
R_X \triangleq \frac{L_X}{L}, \quad
R_Y \triangleq \frac{L_Y}{L}, \quad
R_S \triangleq \frac{L_S}{L}.
\label{Rates}
\end{equation}

A rate tuple $(R_X, R_Y, R_S)$ is said to be achievable if, for sufficiently large blocklength $L$, there exists a secure aggregation scheme satisfying \eqref{messageX} and \eqref{messageY}, such that the correctness constraint \eqref{Correctness} and the relay security constraint \eqref{RelaySecurity} and the server security constraint \eqref{ServerSecurity} are simultaneously satisfied.

The rate region $\mathcal{R}^\ast$ is defined as the closure of the set of all achievable rate tuples $(R_X, R_Y, R_S)$ over sufficiently large $L$.

\section{Main Result}
\label{sec-main}

In this section, we characterize the optimal rate region of the hierarchical secure aggregation problem. 
We identify the fundamental limits imposed by the relay and server security constraints on the two hops, which determine the minimum required groupwise key rate. 
In addition, we show that reliable communication requires $R_X$ and $R_Y$ to exceed one symbol per source symbol. 
The overall performance is therefore governed by the most stringent of these constraints. 
The main result is summarized in the following theorem.
\begin{theorem}
\label{th-main}
Consider the hierarchical secure aggregation problem with $U$ relays and $V$ users per relay. If $G = 1$, the problem is infeasible.
If  $1 < G \leq UV$, the optimal rate region is given by
\begin{align}
\mathcal{R}^* = \left\{ 
\begin{array}{l}
(R_X, R_Y, R_S): R_X \geq 1, 
R_Y \geq 1, \\
R_S \geq  \max\biggm\{
\frac{V}{\binom{UV}{G} - \binom{(U-1)V}{G}},\;
\frac{U - 1}{\binom{UV}{G} - U \binom{V}{G}}
\biggm\}
\end{array}
\right\}
\label{3-4}
\end{align}
\end{theorem}

\begin{remark}\label{remark-1}
The two terms in $R_S$ in Theorem~\ref{th-main} correspond to the constraints imposed by relay security and server security, respectively.

For the relay security constraint, the term 
$\frac{V}{\binom{UV}{G} - \binom{(U-1)V}{G}}$
arises because only those groupwise keys that involve at least one user connected to the considered relay can contribute to protecting its input. 
The quantity $\binom{(U-1)V}{G}$ counts the number of groupwise keys that are formed entirely by users outside this relay, and hence do not provide any protection for it. 
Therefore, these keys must be excluded. 
When $(U-1)V < G$, such groupwise keys do not exist, and the subtraction term vanishes.

For the server security constraint, the term 
$\frac{U - 1}{\binom{UV}{G} - U \binom{V}{G}}$
is due to the requirement that the server should only recover the desired sum while all undesired components are canceled. 
Groupwise keys that are shared only among users within a single relay, counted by $U \binom{V}{G}$, do not create inter-relay coupling and therefore cannot be canceled across relays at the server. 
As a result, these keys must be excluded. 
When $V < G$, such keys do not exist, since any groupwise key must involve users from at least two relays, and the subtraction term disappears.
\end{remark}

\begin{remark}\label{rem:no_feasible_solution_ieee}
According to Theorem~\ref{th-main}, when $G = 1$, the hierarchical secure aggregation problem is infeasible. 
In this case, each groupwise key is held by only a single user, and hence no key is shared across different users. 
As a result, there is no correlation among the keys that can be exploited at the server to cancel undesired components during aggregation. 
Consequently, the server cannot eliminate the interference introduced by the keys, and the correctness constraint in \eqref{Correctness} cannot be satisfied.
\end{remark}

\section{Achievable Scheme}

In this section, we present the proposed secure aggregation scheme. We first illustrate the main ideas through two representative examples, which reveal the key mechanisms of groupwise key evolution under relay security and server security constraints. Based on these examples, we then generalize the construction to arbitrary parameters $U$, $V$, and $G$.

\subsection{Achievability Scheme: Example 1 with $U=2$, $V=2$, and $G=2$}

The model diagram of this example is shown in Fig.~\ref{fig:my_image}. 
In this example, relay security is the dominant constraint. 
We show that the rate tuple $(R_X, R_Y, R_S) = (1,1,2/5)$ is achievable.
Let $L = L_X = L_Y = 5$, i.e., each input $W_{u,v}$ consists of $5$ symbols in $\mathbb{F}_q$, written as
$W_{u,v} = [W_{u,v}(1), W_{u,v}(2), W_{u,v}(3), W_{u,v}(4), W_{u,v}(5)]^{\mathsf{T}}$, 
for $(u,v)\in[2]\times[2]$.
Also, let $L_S = 2$, so that each groupwise key $S_{\mathcal{G}}$ consists of $2$ symbols in $\mathbb{F}_q$, i.e.,
$S_{\mathcal{G}} = [S_{\mathcal{G}}(1), S_{\mathcal{G}}(2)]^{\mathsf{T}}$ 
for all $\mathcal{G}\subset[2]\times[2]$ with $|\mathcal{G}|=2$.

We next specify the transmitted messages from users to the relay. 
Each user encodes its input together with the associated groupwise keys as follows:
\begin{align}
X_{1,1} =& W_{1,1} + \mathbf{H}_{(1,1),(1,2)} S_{(1,1),(1,2)} + \mathbf{H}_{(1,1),(2,1)} S_{(1,1),(2,1)} \notag \\
& + \mathbf{H}_{(1,1),(2,2)} S_{(1,1),(2,2)}, \notag \\
X_{1,2} =& W_{1,2} - \mathbf{H}_{(1,1),(1,2)} S_{(1,1),(1,2)} + \mathbf{H}_{(1,2),(2,1)} S_{(1,2),(2,1)} \notag \\
& + \mathbf{H}_{(1,2),(2,2)} S_{(1,2),(2,2)}, \notag \\
X_{2,1} =& W_{2,1} - \mathbf{H}_{(1,1),(2,1)} S_{(1,1),(2,1)} - \mathbf{H}_{(1,2),(2,1)} S_{(1,2),(2,1)} \notag \\
& + \mathbf{H}_{(2,1),(2,2)} S_{(2,1),(2,2)}, \notag \\
X_{2,2} =& W_{2,2} - \mathbf{H}_{(1,1),(2,2)} S_{(1,1),(2,2)} - \mathbf{H}_{(1,2),(2,2)} S_{(1,2),(2,2)} \notag \\
& - \mathbf{H}_{(2,1),(2,2)} S_{(2,1),(2,2)}. \label{messageX0}
\end{align}
where each $\mathbf{H}_{\mathcal{G}} \in \mathbb{F}_q^{5 \times 2}$ is a precoding matrix used to embed the groupwise keys into the transmitted signals.

The messages received at the relays are then linearly combined as
\begin{align}
   Y_1 =& X_{1,1} + X_{1,2} \notag \\
   =&  W_{1,1} + W_{1,2} + \mathbf{H}_{(1,1),(2,1)} S_{(1,1),(2,1)} + \mathbf{H}_{(1,1),(2,2)}  \notag \\
   &  S_{(1,1),(2,2)} + \mathbf{H}_{(1,2),(2,1)} S_{(1,2),(2,1)} + \mathbf{H}_{(1,2),(2,2)}  \notag \\
   &S_{(1,2),(2,2)} \notag \\
    Y_2 =& X_{2,1} + X_{2,2} \notag \\
    =&  W_{2,1} + W_{2,2} - \mathbf{H}_{(1,1),(2,1)} S_{(1,1),(2,1)} - \mathbf{H}_{(1,1),(2,2)}  \notag \\
    &  S_{(1,1),(2,2)} - \mathbf{H}_{(1,2),(2,1)} S_{(1,2),(2,1)} - \mathbf{H}_{(1,2),(2,2)}  \notag \\
   & S_{(1,2),(2,2)} \label{messageY0}
\end{align}


It can be verified that the carefully designed precoding matrices $\{\mathbf{H}_{\mathcal{G}}\}$ ensure that all groupwise key components cancel at each relay, thereby satisfying relay security, while the desired sums of user inputs are preserved.

This completes the description of the achievable scheme for this example.

\textit{Decodability and construction:}
We now explain the construction of the precoding matrices $\{\mathbf{H}_{\mathcal{G}}\}$ in \eqref{messageX0}. 
The design exploits a structured zero-sum property across users sharing each groupwise key.

For groups with $|\mathcal{G}|=2$, the matrices are assigned in a symmetric manner, where the two users in each group employ $\mathbf{H}_{\mathcal{G}}$ and $-\mathbf{H}_{\mathcal{G}}$, respectively. 
Under this construction, the correctness condition is satisfied since
$\sum_{u\in[2]} Y_u = \sum_{(u,v)\in[2]\times[2]} W_{u,v}.$

To ensure security, the matrices $\{\mathbf{H}_{\mathcal{G}}\}$ are required to be sufficiently generic. 
In the general analysis, we will show that choosing each $\mathbf{H}_{\mathcal{G}}$ from a sufficiently large design space guarantees that the induced linear system satisfies the required independence conditions. 
This can be achieved by increasing the blocklength $L$, which effectively corresponds to operating over a sufficiently large extension field while keeping the base field $\mathbb{F}_q$ fixed.

For illustration, we provide one valid deterministic construction over $\mathbb{F}_5$. 
The precoding matrices are given as follows:
\begin{align}
\setlength{\arraycolsep}{0.2cm} 
\renewcommand{\arraystretch}{1} 
\begin{array}{ccc}
\mathbf{H}_{(1,1),(1,2)} & \mathbf{H}_{(1,1),(2,1)} & \mathbf{H}_{(1,1),(2,2)} \\
{=}\begin{bmatrix} 
1 & 0 \\
0 & 1 \\
1 & 1 \\
1 & 2 \\
2 & 1 
\end{bmatrix},
&
{=}\begin{bmatrix} 
1 & 2 \\
2 & 1 \\
0 & 1 \\
1 & 0 \\
1 & 1 
\end{bmatrix},
&
{=}\begin{bmatrix} 
1 & 1 \\
0 & 2 \\
2 & 0 \\
1 & 2 \\
2 & 1 
\end{bmatrix},
\\[30pt] 
\mathbf{H}_{(1,2),(2,1)} & \mathbf{H}_{(1,2),(2,2)} & \mathbf{H}_{(2,1),(2,2)} \\
{=}\begin{bmatrix} 
2 & 1 \\
1 & 1 \\
1 & 0 \\
0 & 3 \\
2 & 2 
\end{bmatrix},
&
{=}\begin{bmatrix} 
0 & 1 \\
1 & 0 \\
2 & 1 \\
1 & 2 \\
1 & 1 
\end{bmatrix},
&
{=}\begin{bmatrix} 
1 & 0 \\
1 & 1 \\
2 & 2 \\
2 & 1 \\
0 & 2 
\end{bmatrix}.
\end{array}
\label{Evidencedesign1}
\end{align}

It can be verified that this construction ensures that all groupwise key components cancel appropriately across relays while preserving the desired sum of inputs. Moreover, the randomness (or genericity) of the matrices guarantees that no unintended information leakage occurs, thereby satisfying both relay and server security constraints.

\textit{Relay Security:}
The design of the precoding matrices in \eqref{Evidencedesign1} has a key structural property. 
When $G=2$, each user holds $\binom{UV-1}{G-1}=3$ groupwise keys, and the two users connected to each relay jointly hold 
$\binom{UV}{G} - \binom{(U-1)V}{G} = 5$
distinct groupwise keys.

To ensure relay security in \eqref{RelaySecurity}, the inputs $W_{1,1},$ $W_{1,2}, $ $W_{2,1},$ $W_{2,2}$ must remain independent of the received messages $\{X_{u,v}\}_{v\in[2]}$. 
This is achieved by designing the precoding matrix to be full rank, so that the embedded keys fully mask the inputs. 
In particular, since $5L_S = 2L = 10$, the key space matches the dimension of the transmitted signals, ensuring perfect masking.

We now formalize this argument for Relay~$1$. We have
\begin{align}
&I(W_{1,1},W_{1,2},W_{2,1},W_{2,2};X_{1,1},X_{1,2}) \notag\\
=& H(X_{1,1},X_{1,2}) - H(X_{1,1},X_{1,2} \mid W_{1,1},W_{1,2},W_{2,1},W_{2,2})  \label{eq:step12}\\
=& H(X_{1,1},X_{1,2}) \notag\\
&- H(W_{1,1}+Z_{1,1},W_{1,2}+Z_{1,2} \mid W_{1,1},W_{1,2},W_{2,1},W_{2,2})  \label{eq:step12-1}\\
=& H(X_{1,1},X_{1,2}) - H(Z_{1,1},Z_{1,2} \mid W_{1,1},W_{1,2},W_{2,1},W_{2,2})  \label{eq:step12-2}\\
\stackrel{\eqref{independent2}}{=}& H(X_{1,1},X_{1,2}) - H(Z_{1,1},Z_{1,2})  \label{eq:step12-3}\\
\stackrel{\eqref{messageX0}}{=} & 2L -  \\
H&\!\left(
\underbrace{\begin{array}{c@{}c}
  \left[
  \begin{array}{ccc}
    \mathbf{H}_{(1,1),(1,2)} & \mathbf{H}_{(1,1),(2,1)} & \mathbf{H}_{(1,1),(2,2)} \\
    -\mathbf{H}_{(1,1),(1,2)} & 0 & 0 \\
  \end{array}
  \right.
  \\
  \left.
  \begin{array}{ccc}
   0  &0 \\
    \mathbf{H}_{(1,2),(2,1)} & \mathbf{H}_{(1,2),(2,2)} \\
  \end{array}
  \right]
\end{array}}_{\triangleq \widehat{\mathbf{H}}_{10 \times 10}}
\begin{bmatrix}
S_{(1,1),(1,2)} \\
S_{(1,1),(2,1)} \\
S_{(1,1),(2,2)} \\
S_{(1,2),(2,1)} \\
S_{(1,2),(2,2)}
\end{bmatrix}
\right)  \label{eq:step13}\\
= & 2\times 5 - 10=0.  \label{eq:step14}
\end{align}
The second term of \eqref{eq:step12-3} follows from the independence between the inputs and the key variables (see \eqref{independent2}). 

In \eqref{eq:step13}, the entropy term corresponds to the precoded groupwise keys induced by the construction in \eqref{messageX0}. Since the effective precoding matrix $\widehat{\mathbf{H}}_{10 \times 10}$ is full rank (as shown in \eqref{Evidencedesign1}), and the groupwise keys are mutually independent, this entropy term evaluates to $10$.
Substituting into \eqref{eq:step14}, we obtain $I(W_{1,1},W_{1,2},W_{2,1},W_{2,2};X_{1,1},X_{1,2}) = 0$, which establishes the relay security for Relay~$1$. The result for Relay~$2$ follows by symmetry.

We next proceed to establish the server security.

\textit{Server Security:}  
We now establish the server security condition in \eqref{ServerSecurity}.  
When $G = 2$, the total number of groupwise keys shared among the $4$ users is $\binom{UV}{G} = 6$.  
Among these, $U\binom{V}{G} = 2$ keys are shared exclusively within individual relays and therefore do not contribute to inter-relay masking.  
Consequently, for the relay transmissions $Y_1$ and $Y_2$, only $\binom{UV}{G} - U\binom{V}{G} = 4$ groupwise keys are effective in providing masking against the server.

Since the server is allowed to recover the sum $\sum_{(u,v)\in[2]\times[2]} W_{u,v}$, which contains $L$ symbols, only the remaining uncertainty needs to be protected. 
As $4L_S > L$, sufficient randomness is available to ensure server security. 
We now formalize the above intuition.

We have
\begin{align}
&I\Big({W_{1,1},W_{1,2},W_{2,1},W_{2,2}};{Y_1,Y_2}\Big| \sum_{(u,v)\in [2]\times[2]}W_{u,v}\Big) \notag\\
=& H\Big({Y_1,Y_2}\Big| \sum_{(u,v)\in [2]\times[2]}W_{u,v}\Big) - H\Big({Y_1,Y_2}\Big| \notag\\
&\sum_{(u,v)\in [2]\times[2]}W_{u,v},{W_{1,1},W_{1,2},W_{2,1},W_{2,2}}\Big)  \label{MassageY-1}\\
=& H\Big({X_{1,1}+X_{1,2},X_{2,1}+X_{2,2}}\Big| \sum_{(u,v)\in [2]\times[2]}W_{u,v}\Big) \notag\\
& - H({X_{1,1}+X_{1,2},X_{2,1}+X_{2,2}} |  {W_{1,1},W_{1,2},W_{2,1},W_{2,2}})  \label{MassageY-2}\\
=& H\Big(W_{1,1}+W_{1,2}+Z_{1,1}+Z_{1,2},W_{2,1}+W_{2,2}+Z_{2,1}+\notag\\
&Z_{2,2}\Big| \sum_{(u,v)\in [2]\times[2]}W_{u,v}\Big)  - H(W_{1,1}+W_{1,2}+Z_{1,1}+Z_{1,2},\notag\\
&W_{2,1}+W_{2,2}+Z_{2,1} +Z_{2,2} |  {W_{1,1},W_{1,2},W_{2,1},W_{2,2}} )  \label{MassageY-3}\\
=& H\Big(W_{1,1}+W_{1,2}+Z_{1,1}+Z_{1,2},W_{2,1}+W_{2,2}+Z_{2,1}+\notag\\
&Z_{2,2}, \sum_{(u,v)\in [2]\times[2]}W_{u,v}\Big)-H\Big(\sum_{(u,v)\in [2]\times[2]}W_{u,v}\Big) \notag\\
& - H({Z_{1,1}+Z_{1,2},Z_{2,1}+Z_{2,2}} |  {W_{1,1},W_{1,2},W_{2,1},W_{2,2}})  \label{MassageY-4}\\
\stackrel{\eqref{independent2}}{=}&  H(W_{1,1}+W_{1,2}+Z_{1,1}+Z_{1,2},W_{2,1}+W_{2,2}+Z_{2,1}+\notag\\
&Z_{2,2})-H\Big(\sum_{(u,v)\in [2]\times[2]}W_{u,v}\Big)  - H(Z_{1,1}+Z_{1,2})  \label{MassageY-5}\\
\stackrel{\eqref{messageX0}}{=} &2\times 5-5 \notag\\
&- H{
\left(
\underbrace{\begin{array}{c@{}c}
  \left[
  \begin{array}{ccc}
    \mathbf{H}_{(1,1),(2,1)} & \mathbf{H}_{(1,1),(2,2)} 
  \end{array}
  \right.
  \\
  \left.
  \begin{array}{ccc}
  \mathbf{H}_{(1,2),(2,1)}  & \mathbf{H}_{(1,2),(2,2)}
  \end{array}
  \right]
\end{array}}_{\triangleq \widehat{\mathbf{H}}_{5 \times 10}}
\begin{bmatrix}
S_{(1,1),(2,1)} \\
S_{(1,1),(2,2)} \\
S_{(1,2),(2,1)} \\
S_{(1,2),(2,2)} \\
\end{bmatrix}
\right)
}  \label{MassageY-6} \\
= &5 - 5=0.  \label{MassageY-7} 
\end{align}
In \eqref{MassageY-5}, the first term follows from the correctness constraint in \eqref{Correctness}, which ensures that the sum $\sum_{(u,v)\in [2]\times[2]}W_{u,v}$ is a deterministic function of $W_{1,1}+W_{1,2}+Z_{1,1}+Z_{1,2}$ and $W_{2,1}+W_{2,2}+Z_{2,1}+Z_{2,2}$.  
The third term follows from the independence between the input and key variables (see \eqref{independent2}). Moreover, due to the zero-sum structure of the key design, we have $Z_{1,1}+Z_{1,2}+Z_{2,1}+Z_{2,2}=0$, which implies that $Z_{2,1}+Z_{2,2}$ is fully determined by $Z_{1,1}+Z_{1,2}$.

In \eqref{MassageY-6}, the entropy term corresponds to the precoded groupwise keys induced by the construction in \eqref{messageX0}. Since the effective precoding matrix $\widehat{\mathbf{H}}_{5 \times 10}$ is of rank $5$ (as shown in \eqref{Evidencedesign1}), and the groupwise keys are mutually independent, this term evaluates to $5$.

Therefore, the mutual information in \eqref{MassageY-7} evaluates to zero, which establishes the server security constraint.

\subsection{Achievability Scheme: Example 2 with $U=4$, $V=2$, and $G=7$}

In this example, server security is the dominant constraint. 
We show that the rate tuple $(R_X, R_Y, R_S) = (1,1,3/8)$ is achievable.
Let $L = L_X = L_Y = 8$, i.e., each input $W_{u,v}$ consists of $8$ symbols over $\mathbb{F}_q$, written as
$W_{u,v} = [W_{u,v}(1), W_{u,v}(2), W_{u,v}(3), W_{u,v}(4), W_{u,v}(5),$ $ W_{u,v}(6), W_{u,v}(7), W_{u,v}(8)]^{\mathsf{T}},$
for $(u,v)\in[4]\times[2]$.
Also, let $L_S = 3$, so that each groupwise key $S_{\mathcal{G}}$ consists of $3$ symbols over $\mathbb{F}_q$, i.e.,
$S_{\mathcal{G}} = [S_{\mathcal{G}}(1), S_{\mathcal{G}}(2), S_{\mathcal{G}}(3)]^{\mathsf{T}},$
for all $\mathcal{G}\subset[4]\times[2]$ with $|\mathcal{G}|=7$.


For simplicity, we enumerate all subsets $\mathcal{G} \subset [4]\times[2]$ with $|\mathcal{G}|=7$ as $\{\mathcal{G}_i\}_{i=1}^8$, given by
\begin{align*}
\mathcal{G}_1 &= \{(1,1), (1,2), (2,1), (2,2), (3,1), (3,2), (4,1)\},\\
\mathcal{G}_2 &= \{(1,1), (1,2), (2,1), (2,2), (3,1), (3,2), (4,2)\},\\
\mathcal{G}_3 &= \{(1,1), (1,2), (2,1), (2,2), (3,1), (4,1), (4,2)\},\\
\mathcal{G}_4 &= \{(1,1), (1,2), (2,1), (2,2), (3,2), (4,1), (4,2)\},\\
\mathcal{G}_5 &= \{(1,1), (1,2), (2,1), (3,1), (3,2), (4,1), (4,2)\},\\
\mathcal{G}_6 &= \{(1,1), (1,2), (2,2), (3,1), (3,2), (4,1), (4,2)\},\\
\mathcal{G}_7 &= \{(1,1), (2,1), (2,2), (3,1), (3,2), (4,1), (4,2)\},\\
\mathcal{G}_8 &= \{(1,2), (2,1), (2,2), (3,1), (3,2), (4,1), (4,2)\}.
\end{align*}

We next specify the messages transmitted from users to the relays. 
Each user encodes its input together with the associated groupwise keys as follows:
\begin{align}
X_{1,1} &= W_{1,1} + \mathbf{H}_{\mathcal{G}_1}^{1,1}S_{\mathcal{G}_1} + \mathbf{H}_{\mathcal{G}_2}^{1,1}S_{\mathcal{G}_2}+ \mathbf{H}_{\mathcal{G}_3}^{1,1}S_{\mathcal{G}_3} +  \mathbf{H}_{\mathcal{G}_4}^{1,1}S_{\mathcal{G}_4}+ \notag\\
&\quad + \mathbf{H}_{\mathcal{G}_5}^{1,1} S_{\mathcal{G}_5} + \mathbf{H}_{\mathcal{G}_6}^{1,1} S_{\mathcal{G}_6}+ \mathbf{H}_{\mathcal{G}_7}^{1,1} S_{\mathcal{G}_7} \notag\\
X_{1,2} &= W_{1,2} + \mathbf{H}_{\mathcal{G}_1}^{1,2}S_{\mathcal{G}_1} + \mathbf{H}_{\mathcal{G}_2}^{1,2}S_{\mathcal{G}_2}+ \mathbf{H}_{\mathcal{G}_3}^{1,2}S_{\mathcal{G}_3} +  \mathbf{H}_{\mathcal{G}_4}^{1,2}S_{\mathcal{G}_4}+ \notag\\
&\quad + \mathbf{H}_{\mathcal{G}_5}^{1,2} S_{\mathcal{G}_5} + \mathbf{H}_{\mathcal{G}_6}^{1,2} S_{\mathcal{G}_6}+ \mathbf{H}_{\mathcal{G}_8}^{1,2} S_{\mathcal{G}_8} \notag\\
X_{2,1} &= W_{2,1} + \mathbf{H}_{\mathcal{G}_1}^{2,1}S_{\mathcal{G}_1} + \mathbf{H}_{\mathcal{G}_2}^{2,1}S_{\mathcal{G}_2}+ \mathbf{H}_{\mathcal{G}_3}^{2,1}S_{\mathcal{G}_3} +  \mathbf{H}_{\mathcal{G}_4}^{2,1}S_{\mathcal{G}_4}+ \notag \\
&\quad + \mathbf{H}_{\mathcal{G}_5}^{2,1} S_{\mathcal{G}_5} + \mathbf{H}_{\mathcal{G}_7}^{2,1}S_{\mathcal{G}_7}+ \mathbf{H}_{\mathcal{G}_8}^{2,1} S_{\mathcal{G}_8} \notag\\
X_{2,2} &= W_{2,2} + \mathbf{H}_{\mathcal{G}_1}^{2,2}S_{\mathcal{G}_1} + \mathbf{H}_{\mathcal{G}_2}^{2,2}S_{\mathcal{G}_2}+ \mathbf{H}_{\mathcal{G}_3}^{2,2}S_{\mathcal{G}_3} +  \mathbf{H}_{\mathcal{G}_4}^{2,2}S_{\mathcal{G}_4}+ \notag\\
&\quad + \mathbf{H}_{\mathcal{G}_6}^{2,2} S_{\mathcal{G}_6} + \mathbf{H}_{\mathcal{G}_7}^{2,2}S_{\mathcal{G}_7}+ \mathbf{H}_{\mathcal{G}_8}^{2,2} S_{\mathcal{G}_8} \notag\\
X_{3,1} &= W_{3,1} + \mathbf{H}_{\mathcal{G}_1}^{3,1}S_{\mathcal{G}_1} + \mathbf{H}_{\mathcal{G}_2}^{3,1}S_{\mathcal{G}_2}+ \mathbf{H}_{\mathcal{G}_3}^{3,1}S_{\mathcal{G}_3} +  \mathbf{H}_{\mathcal{G}_5}^{3,1}S_{\mathcal{G}_5}+ \notag \\
&\quad + \mathbf{H}_{\mathcal{G}_6}^{3,1} S_{\mathcal{G}_6} + \mathbf{H}_{\mathcal{G}_7}^{3,1}S_{\mathcal{G}_7}+ \mathbf{H}_{\mathcal{G}_8}^{3,1} S_{\mathcal{G}_8} \notag\\
X_{3,2} &= W_{3,2} + \mathbf{H}_{\mathcal{G}_1}^{3,2}S_{\mathcal{G}_1} + \mathbf{H}_{\mathcal{G}_2}^{3,2}S_{\mathcal{G}_2}+ \mathbf{H}_{\mathcal{G}_4}^{3,2}S_{\mathcal{G}_4} +  \mathbf{H}_{\mathcal{G}_5}^{3,2}S_{\mathcal{G}_5}+ \notag\\
&\quad + \mathbf{H}_{\mathcal{G}_6}^{3,2} S_{\mathcal{G}_6} + \mathbf{H}_{\mathcal{G}_7}^{3,2}S_{\mathcal{G}_7}+ \mathbf{H}_{\mathcal{G}_8}^{3,2}S_{\mathcal{G}_8} \notag\\
X_{4,1} &= W_{4,1} + \mathbf{H}_{\mathcal{G}_1}^{4,1}S_{\mathcal{G}_1} + \mathbf{H}_{\mathcal{G}_3}^{4,1}S_{\mathcal{G}_3}+ \mathbf{H}_{\mathcal{G}_4}^{4,1}S_{\mathcal{G}_4} +  \mathbf{H}_{\mathcal{G}_5}^{4,1}S_{\mathcal{G}_5}+ \notag\\
&\quad + \mathbf{H}_{\mathcal{G}_6}^{4,1} S_{\mathcal{G}_6} + \mathbf{H}_{\mathcal{G}_7}^{4,1}S_{\mathcal{G}_7}+ \mathbf{H}_{\mathcal{G}_8}^{4,1} S_{\mathcal{G}_8} \notag\\
X_{4,2} &= W_{4,2} + \mathbf{H}_{\mathcal{G}_2}^{4,2}S_{\mathcal{G}_2} + \mathbf{H}_{\mathcal{G}_3}^{4,2}S_{\mathcal{G}_3} + \mathbf{H}_{\mathcal{G}_4}^{4,2}S_{\mathcal{G}_4} + \mathbf{H}_{\mathcal{G}_5}^{4,2}S_{\mathcal{G}_5}+ \notag\\
&\quad + \mathbf{H}_{\mathcal{G}_6}^{4,2} S_{\mathcal{G}_6} + \mathbf{H}_{\mathcal{G}_7}^{4,2}S_{\mathcal{G}_7}+\mathbf{H}_{\mathcal{G}_8}^{4,2} S_{\mathcal{G}_8}
\label{MessageX1}
\end{align}
where each $\mathbf{H}_{\mathcal{G}_i}^{u,v} \in \mathbb{F}_q^{8 \times 3}$ denotes a precoding matrix used to embed the corresponding groupwise key $S_{\mathcal{G}_i}$ into the transmitted signals.

The messages received at each relay are then linearly combined to generate the relay transmissions. Specifically, each relay sums the messages from its associated users, yielding
\begin{equation}
\begin{aligned}
Y_1 =& X_{1,1} + X_{1,2} \\
    =&  W_{1,1} + W_{1,2} + (\mathbf{H}_{\mathcal{G}_1}^{1,1}+ \mathbf{H}_{\mathcal{G}_1}^{1,2})S_{\mathcal{G}_1} + (\mathbf{H}_{\mathcal{G}_2}^{1,1}+ \mathbf{H}_{\mathcal{G}_2}^{1,2})S_{\mathcal{G}_2} \\
    &+ (\mathbf{H}_{\mathcal{G}_3}^{1,1}+ \mathbf{H}_{\mathcal{G}_3}^{1,2})S_{\mathcal{G}_3} +  
    (\mathbf{H}_{\mathcal{G}_4}^{1,1}+ \mathbf{H}_{\mathcal{G}_4}^{1,2})S_{\mathcal{G}_4}+ (\mathbf{H}_{\mathcal{G}_5}^{1,1}+\\
    & \mathbf{H}_{\mathcal{G}_5}^{1,2}) S_{\mathcal{G}_5} + (\mathbf{H}_{\mathcal{G}_6}^{1,1}+ \mathbf{H}_{\mathcal{G}_6}^{1,2})S_{\mathcal{G}_6} + \mathbf{H}_{\mathcal{G}_7}^{1,1}S_{\mathcal{G}_7}+ \mathbf{H}_{\mathcal{G}_8}^{1,2} S_{\mathcal{G}_8}\\
Y_2 =& X_{2,1} + X_{2,2} \\
    =&  W_{2,1} + W_{2,2} + (\mathbf{H}_{\mathcal{G}_1}^{2,1}+ \mathbf{H}_{\mathcal{G}_1}^{2,2})S_{\mathcal{G}_1} + (\mathbf{H}_{\mathcal{G}_2}^{2,1}+ \mathbf{H}_{\mathcal{G}_2}^{2,2})S_{\mathcal{G}_2} \\
    &+ (\mathbf{H}_{\mathcal{G}_3}^{2,1}+ \mathbf{H}_{\mathcal{G}_3}^{2,2})S_{\mathcal{G}_3} +  
    (\mathbf{H}_{\mathcal{G}_4}^{2,1}+ \mathbf{H}_{\mathcal{G}_4}^{2,2})S_{\mathcal{G}_4}+ \mathbf{H}_{\mathcal{G}_5}^{2,1}S_{\mathcal{G}_5}\\
    &  +  \mathbf{H}_{\mathcal{G}_6}^{2,2}S_{\mathcal{G}_6} + (\mathbf{H}_{\mathcal{G}_7}^{2,1}+ \mathbf{H}_{\mathcal{G}_7}^{2,2})S_{\mathcal{G}_7}+ (\mathbf{H}_{\mathcal{G}_8}^{2,1}+ \mathbf{H}_{\mathcal{G}_8}^{2,2}) S_{\mathcal{G}_8}\\
Y_3 =& X_{3,1} + X_{3,2} \\
    =&  W_{3,1} + W_{3,2} + (\mathbf{H}_{\mathcal{G}_1}^{3,1}+ \mathbf{H}_{\mathcal{G}_1}^{3,2})S_{\mathcal{G}_1} + (\mathbf{H}_{\mathcal{G}_2}^{3,1}+ \mathbf{H}_{\mathcal{G}_2}^{3,2})S_{\mathcal{G}_2} \\
    &+ \mathbf{H}_{\mathcal{G}_3}^{3,1}S_{\mathcal{G}_3} + \mathbf{H}_{\mathcal{G}_4}^{3,2} S_{\mathcal{G}_4}+ (\mathbf{H}_{\mathcal{G}_5}^{3,1}+ \mathbf{H}_{\mathcal{G}_5}^{3,2})S_{\mathcal{G}_5}  +  (\mathbf{H}_{\mathcal{G}_6}^{3,1}+ \\ &\mathbf{H}_{\mathcal{G}_6}^{3,2})S_{\mathcal{G}_6} + (\mathbf{H}_{\mathcal{G}_7}^{3,1}+ \mathbf{H}_{\mathcal{G}_7}^{3,2})S_{\mathcal{G}_7}+ (\mathbf{H}_{\mathcal{G}_8}^{3,1}+ \mathbf{H}_{\mathcal{G}_8}^{3,2}) S_{\mathcal{G}_8}\\
Y_4 =& X_{4,1} + X_{4,2} \\
    =&  W_{4,1} + W_{4,2} + \mathbf{H}_{\mathcal{G}_1}^{4,1}S_{\mathcal{G}_1} +  \mathbf{H}_{\mathcal{G}_2}^{4,2}S_{\mathcal{G}_2} 
    +(\mathbf{H}_{\mathcal{G}_3}^{4,1}+\mathbf{H}_{\mathcal{G}_3}^{4,2})  \\
    &S_{\mathcal{G}_3}+(\mathbf{H}_{\mathcal{G}_4}^{4,1}+ \mathbf{H}_{\mathcal{G}_4}^{4,2}) S_{\mathcal{G}_4}+ (\mathbf{H}_{\mathcal{G}_5}^{4,1}+ \mathbf{H}_{\mathcal{G}_5}^{4,2})S_{\mathcal{G}_5}  +  (\mathbf{H}_{\mathcal{G}_6}^{4,1} \\ &+\mathbf{H}_{\mathcal{G}_6}^{4,2})S_{\mathcal{G}_6} + (\mathbf{H}_{\mathcal{G}_7}^{4,1}+ \mathbf{H}_{\mathcal{G}_7}^{4,2})S_{\mathcal{G}_7}+ (\mathbf{H}_{\mathcal{G}_8}^{4,1}+ \mathbf{H}_{\mathcal{G}_8}^{4,2}) S_{\mathcal{G}_8}
\end{aligned}
\label{messageY1}
\end{equation}

\textit{Decodability and construction:}
We now explain the construction of the precoding matrices $\{\mathbf{H}_{\mathcal{G}_i}^{u,v}\}$ in \eqref{MessageX1}. 
The design exploits a structured zero-sum property across users sharing each groupwise key.
For each group $\mathcal{G}_i$ with $|\mathcal{G}_i|=7$, the precoding matrices satisfy the following zero-sum condition:
\begin{equation}
\sum_{(u,v)\in \mathcal{G}_i} \mathbf{H}_{\mathcal{G}_i}^{u,v} = \mathbf{0}, \forall i\in [8].
\label{zerosum}
\end{equation}
That is, the contributions of the same groupwise key across all participating users cancel out when aggregated.
Specifically, for each $\mathcal{G}_i$, we randomly generate six $8\times 3$ matrices $\mathbf{H}_{\mathcal{G}_i}^{u,v}$, and set the remaining one as
\begin{equation}
\mathbf{H}_{\mathcal{G}_i}^{u',v'}
= - \sum_{(u,v)\in \mathcal{G}_i\setminus \{(u',v')\}} \mathbf{H}_{\mathcal{G}_i}^{u,v}, \forall~(u',v') \in \mathcal{G}_i.
\label{Co_constraint}
\end{equation}
Under this construction, the correctness condition is satisfied since
\begin{equation}
\sum_{u\in[4]} Y_u 
= \sum_{(u,v)\in[4]\times[2]} W_{u,v},
\end{equation}
as all groupwise key contributions vanish due to \eqref{zerosum}.

To ensure security, the matrices $\{\mathbf{H}_{\mathcal{G}_i}^{u,v}\}$ are required to be sufficiently generic. 
In the general analysis, we will show that selecting each $\mathbf{H}_{\mathcal{G}}$ from a sufficiently large design space guarantees that the induced linear system satisfies the required independence conditions. 
This can be achieved by increasing the blocklength $L$, which effectively corresponds to operating over a sufficiently large extension field while keeping the base field $\mathbb{F}_q$ fixed. 
In particular, the construction combines Vandermonde-type structures to obtain matrices that satisfy the desired properties.

For illustration, we provide one valid deterministic construction over $q=11$. 
The precoding matrices are constructed based on a Vandermonde-type structure.

Specifically, for each $i\in[1:8]$, we construct $\mathbf{H}_{\mathcal{G}_i}^{u,v} \in \mathbb{F}_q^{8\times 3}$ as
\begin{align}
\mathbf{H}_{\mathcal{G}_i}^{u,v} = 
\begin{bmatrix}
(b_{i,1})^{e_{u,v}} & (b_{i,2})^{e_{u,v}} & (b_{i,3})^{e_{u,v}} \\
(b_{i,1})^{e_{u,v}+1} & (b_{i,2})^{e_{u,v}+1} & (b_{i,3})^{e_{u,v}+1} \\
\vdots & \vdots & \vdots \\
(b_{i,1})^{e_{u,v}+7} & (b_{i,2})^{e_{u,v}+7} & (b_{i,3})^{e_{u,v}+7}
\end{bmatrix}.
\label{Construct_matrix}
\end{align}

If $(u,v)\notin \mathcal{G}_i$, we set $\mathbf{H}_{\mathcal G_i}^{u,v}=0_{8\times3}$. 
All operations are performed over the finite field $\mathbb{F}_{11}$. Let $g=2$ be a primitive element. 
For each $i\in[1:8]$, define three bases as
\begin{align}
b_{i,1} = g^{\,i-1}, \quad
b_{i,2} = g^{\,i+2}, \quad
b_{i,3} = g^{\,i+5} \pmod{10}.
\label{Element_set}
\end{align}
The modulo-$10$ operation ensures that the exponents are taken over the multiplicative group of $\mathbb{F}_{11}$.

For $(u, v) \in [4] \times [2]$, we assign the exponents as
\begin{align}
&e_{1,1} = 0,\ e_{1,2} = 4, 
e_{2,1} = 1,\ e_{2,2} = 5, \notag\\ 
&e_{3,1} = 2,\ e_{3,2} = 6, 
e_{4,1} = 3.
\label{Exponent_set}
\end{align}

With the above choice of bases and exponents, the matrices $\mathbf{H}_{\mathcal{G}_i}^{u,v}$ inherit a Vandermonde-type structure, which ensures sufficient linear independence required for security.

Finally, to satisfy the zero-sum constraint in \eqref{Co_constraint} when $|\mathcal{G}_i|=7$, we construct the remaining matrices as
\begin{align}
\mathbf{H}_{\mathcal G_1}^{4,1} &= -\sum_{u=1}^{3} \left( \mathbf{H}_{\mathcal G_1}^{u,1} + \mathbf{H}_{\mathcal G_1}^{u,2} \right),\notag\\
\mathbf{H}_{\mathcal G_i}^{4,2} &= -\sum_{(u,v)\in[4]\times[2]\setminus\{(4,2)\}} \mathbf{H}_{\mathcal G_i}^{u,v}, \quad i\in[2:8].
\label{Evidencedesign2}
\end{align}

This completes the construction of the precoding matrices $\mathbf{H}_{\mathcal{G}_i}^{u,v}$.

\textit{Relay Security:} 
From \eqref{Construct_matrix} to \eqref{Evidencedesign2}, the precoding design exhibits a crucial structural property. When $G = 7$, each user possesses $\binom{UV-1}{G-1} = 7$ groupwise keys, while the two users connected to each relay jointly possess $\binom{UV}{G} - \binom{(U-1)V}{G} = 8$ groupwise keys, covering all groupwise keys involved in the construction.

To ensure relay security in \eqref{RelaySecurity}, the input sequence $W_{[4]\times[2]}$ must be independent of the received message sequence $\{X_{u,v}\}_{v\in[2]}$. This is achieved by designing the precoding matrix to be full-rank, so that the embedded keys fully mask the input information. Since $8L_S=24 > 2L = 16$, the available key space is sufficient to match the dimension of the transmitted signals, ensuring perfect masking. Therefore, relay security is achieved.

We now formalize this argument for Relay 1. We have
\begin{align}
&I(W_{[4]\times[2]};X_{1,1},X_{1,2}) \notag\\
=& H(X_{1,1},X_{1,2}) - H(X_{1,1},X_{1,2} \mid W_{[4]\times[2]})  \label{eq:step112}\\
=& H(X_{1,1},X_{1,2}) \notag\\
&- H(W_{1,1}+Z_{1,1},W_{1,2}+Z_{1,2} \mid W_{[4]\times[2]})  \label{eq:step112-1}\\
=& H(X_{1,1},X_{1,2}) - H(Z_{1,1},Z_{1,2} \mid W_{[4]\times[2]})  \label{eq:step112-2}\\
\stackrel{\eqref{independent2}}{=}& H(X_{1,1},X_{1,2}) - H(Z_{1,1},Z_{1,2})  \label{eq:step112-3}\\
\stackrel{\eqref{MessageX1}}{=} & 2L -
H{
\left(
\underbrace{\begin{array}{c@{}c}
  \left[
  \begin{array}{cccc}
    \mathbf{H}_{\mathcal{G}_1}^{1,1} & \mathbf{H}_{\mathcal{G}_2}^{1,1}  & \mathbf{H}_{\mathcal{G}_3}^{1,1}  & \mathbf{H}_{\mathcal{G}_4}^{1,1}    \\[4pt]
   \mathbf{H}_{\mathcal{G}_1}^{1,2}  & \mathbf{H}_{\mathcal{G}_2}^{1,2} & \mathbf{H}_{\mathcal{G}_3}^{1,2} & \mathbf{H}_{\mathcal{G}_4}^{1,2}   \\
  \end{array}
  \right.
  \\[12pt]
  \left.
  \begin{array}{cccc}
   \mathbf{H}_{\mathcal{G}_5}^{1,1}  & \mathbf{H}_{\mathcal{G}_6}^{1,1}  & \mathbf{H}_{\mathcal{G}_7}^{1,1}  & 0 \\[4pt]
    \mathbf{H}_{\mathcal{G}_5}^{1,2} & \mathbf{H}_{\mathcal{G}_6}^{1,2} & 0 & \mathbf{H}_{\mathcal{G}_8}^{1,2}
  \end{array}
  \right]
\end{array}}_{\triangleq \widehat{\mathbf{H}}_{16 \times 24}}
\begin{bmatrix}
S_{\mathcal{G}_1} \\
S_{\mathcal{G}_2} \\
S_{\mathcal{G}_3} \\
S_{\mathcal{G}_4} \\
S_{\mathcal{G}_5} \\
S_{\mathcal{G}_6} \\
S_{\mathcal{G}_7} \\
S_{\mathcal{G}_8} 
\end{bmatrix}
\right)
}  \label{eq:step113}\\
= & 2\times 8 - 16=0.  \label{eq:step114}
\end{align}
The second term in \eqref{eq:step112-3} follows from the independence between the input and key variables as stated in \eqref{independent2}. In \eqref{eq:step113}, the entropy term corresponds to the precoded groupwise keys constructed in \eqref{MessageX1}. Since the precoding matrix $\widehat{\mathbf{H}}_{16 \times 24}$ has full row rank (as shown in \eqref{Construct_matrix} to \eqref{Evidencedesign2}), the entropy equals $16$ under the independence of the groupwise keys. The final step follows from the independence of groupwise keys (see \eqref{gruopwisesize}). Therefore, Relay $1$ satisfies the security constraint. By symmetry, all other relays also satisfy security. Next, we proceed to the server security analysis.

We now proceed to server security.

\textit{Server Security:} 
We now analyze the server security condition in \eqref{ServerSecurity}. When $G = 7$, there are $\binom{UV}{G} = 8$ groupwise keys in total. These keys are all effective in protecting the second-hop transmission since no groupwise key is confined within a single relay.
Since the server can recover $\sum_{(u,v)\in[4]\times[2]} W_{u,v}$, which contains $L$ symbols, only the remaining uncertainty needs to be protected. Given that $8L_S = 3L = 24$, the available randomness is sufficient to guarantee server security.

We now formalize this intuition. We have
\begin{align}
& I\bigg(W_{[4]\times[2]};\{Y_u\}_{u\in[4]} \bigg| \sum_{(u,v)\in[4]\times[2]} W_{u,v}\bigg) \notag\\
=& H\bigg(\{Y_u\}_{u\in[4]} \bigg| \sum_{(u,v)\in[4]\times[2]} W_{u,v}\bigg) - 
H\bigg(\{Y_u\}_{u\in[4]} \bigg| \notag\\
&\sum_{(u,v)\in[4]\times[2]} W_{u,v},W_{[4]\times[2]}\bigg)
  \label{eq:step30}\\
=& H\bigg(\{X_{u,1}+X_{u,2}\}_{u\in[4]} \bigg| \sum_{(u,v)\in[4]\times[2]} W_{u,v}\bigg)\notag\\
 &-H\bigg(\{X_{u,1}+X_{u,2}\}_{u\in[4]} \bigg| W_{[4]\times[2]}\bigg)
  \label{eq:step30-1}\\  
=& H\bigg(\bigg\{\sum_{v\in[2]}(W_{u,v}+Z_{u,v})\bigg\}_{u\in[4]} \bigg| \sum_{(u,v)\in[4]\times[2]} W_{u,v}\bigg) \notag\\
&- H\bigg(\bigg\{\sum_{v\in[2]}(W_{u,v}+Z_{u,v})\bigg\}_{u\in[4]} \bigg|  W_{[4]\times[2]}\bigg)
  \label{eq:step30-2}\\ 
=& H\bigg(\bigg\{\sum_{v\in[2]}(W_{u,v}+Z_{u,v})\bigg\}_{u\in[4]} , \sum_{(u,v)\in[4]\times[2]} W_{u,v}\bigg)   \notag\\
& -H\bigg( \sum_{(u,v)\in[4]\times[2]} W_{u,v}\bigg) \notag\\
&-H\bigg(\{Z_{u,1}+Z_{u,2}\}_{u\in[4]} \bigg| 
W_{[4]\times[2]}\bigg)
  \label{eq:step30-3}\\ 
\stackrel{\eqref{independent2}}{=}&  H(\{W_{u,1}+W_{u,2}+Z_{u,1}+Z_{u,2}\}_{u\in[4]} ) - \notag \\
&H\bigg( \sum_{(u,v)\in[4]\times[2]} W_{u,v}\bigg)-
H(\{Z_{u,1}+Z_{u,2}\}_{u\in[3]} )
  \label{eq:step30-4}\\ 
\stackrel{\eqref{MessageX1}}{\leq}& 4\times 8 - 8 - H\left(
  {
   \widehat{\mathbf{H}}_{24 \times 24}}
  \begin{bmatrix}
    S_{\mathcal{G}_1} \\
    S_{\mathcal{G}_2} \\
    S_{\mathcal{G}_3} \\
    S_{\mathcal{G}_4} \\
    S_{\mathcal{G}_5} \\
    S_{\mathcal{G}_6} \\
    S_{\mathcal{G}_7} \\ 
    S_{\mathcal{G}_8} 
  \end{bmatrix}
\right)  
\label{eq:step31} \\
=& 24 - 24 = 0 \label{eq:step33}
\end{align}

\begin{figure*}[tb]         
    \hrulefill             
    \vspace{8pt}           
    \begin{align}
    &\widehat{\mathbf{H}}_{24 \times 24\notag}\\
    =&
\begin{bmatrix}
\mathbf{H}_{\mathcal{G}_1}^{1,1}+\mathbf{H}_{\mathcal{G}_1}^{1,2} & \mathbf{H}_{\mathcal{G}_2}^{1,1}+\mathbf{H}_{\mathcal{G}_2}^{1,2} & \mathbf{H}_{\mathcal{G}_3}^{1,1}+\mathbf{H}_{\mathcal{G}_3}^{1,2} & \mathbf{H}_{\mathcal{G}_4}^{1,1}+\mathbf{H}_{\mathcal{G}_4}^{1,2} & \mathbf{H}_{\mathcal{G}_5}^{1,1}+\mathbf{H}_{\mathcal{G}_5}^{1,2} & \mathbf{H}_{\mathcal{G}_6}^{1,1}+\mathbf{H}_{\mathcal{G}_6}^{1,2} &
\mathbf{H}_{\mathcal{G}_7}^{1,1} & \mathbf{H}_{\mathcal{G}_8}^{1,2}\\[4pt]
\mathbf{H}_{\mathcal{G}_1}^{2,1}+\mathbf{H}_{\mathcal{G}_1}^{2,2} & \mathbf{H}_{\mathcal{G}_2}^{2,1}+\mathbf{H}_{\mathcal{G}_2}^{2,2} & \mathbf{H}_{\mathcal{G}_3}^{2,1}+\mathbf{H}_{\mathcal{G}_3}^{2,2} & \mathbf{H}_{\mathcal{G}_4}^{2,1}+\mathbf{H}_{\mathcal{G}_4}^{2,2} & \mathbf{H}_{\mathcal{G}_5}^{2,1} & \mathbf{H}_{\mathcal{G}_6}^{2,2} &
\mathbf{H}_{\mathcal{G}_7}^{2,1}+\mathbf{H}_{\mathcal{G}_7}^{2,2} & \mathbf{H}_{\mathcal{G}_8}^{2,1}+\mathbf{H}_{\mathcal{G}_8}^{2,2}\\[4pt]
\mathbf{H}_{\mathcal{G}_1}^{3,1}+\mathbf{H}_{\mathcal{G}_1}^{3,2} & \mathbf{H}_{\mathcal{G}_2}^{3,1}+\mathbf{H}_{\mathcal{G}_2}^{3,2} & \mathbf{H}_{\mathcal{G}_3}^{3,1} & \mathbf{H}_{\mathcal{G}_4}^{3,2} & \mathbf{H}_{\mathcal{G}_5}^{3,1}+\mathbf{H}_{\mathcal{G}_5}^{3,2} & \mathbf{H}_{\mathcal{G}_6}^{3,1}+\mathbf{H}_{\mathcal{G}_6}^{3,2} &
\mathbf{H}_{\mathcal{G}_7}^{3,1}+\mathbf{H}_{\mathcal{G}_7}^{3,2} & \mathbf{H}_{\mathcal{G}_8}^{3,1}+\mathbf{H}_{\mathcal{G}_8}^{3,2}
    \end{bmatrix}
\label{matrix_H4}
\end{align}
\end{figure*}


Where the matrix $\widehat{\mathbf{H}}_{24 \times 24}$ is given in \eqref{matrix_H4}, shown at the top of the next page.
In \eqref{eq:step30-4}, the first term follows from the correctness constraint in \eqref{Correctness}, which ensures that the sum $\sum_{(u,v)\in [4]\times[2]}W_{u,v}$ is a deterministic function of $\big\{\sum_{v\in[2]}(W_{u,v}+Z_{u,v})\big\}_{u\in[4]}$.
The third term in \eqref{eq:step30-4} follows from the independence between the input and key variables, as stated in \eqref{independent2}. Moreover, due to the zero-sum structure in the key design, we have
$\sum_{u\in[4]} (Z_{u,1} + Z_{u,2}) = 0,$
which implies that the four random variables $\{Z_{u,1}+Z_{u,2}\}_{u\in[4]}$ are linearly dependent, i.e., they span a subspace of dimension 3.

The third term in \eqref{eq:step31} corresponds to the entropy of the effective precoded groupwise keys as defined in \eqref{MessageX1}. Since the rank of the effective precoding matrix $\widehat{\mathbf{H}}_{24 \times 24}$ is $24$ (as shown in \eqref{Construct_matrix} to \eqref{Evidencedesign2}) and all groupwise keys are mutually independent, this term equals $24$.

Consequently, the mutual information in \eqref{eq:step33} is zero, which establishes the server security constraint.

\subsection{General Achievability for Arbitrary $U$, $V$, and $G$}
\label{sec:GENERAL SCHEME}

In this section, we present a general achievable scheme for arbitrary $U$, $V$, and $G$. 
We consider two regimes depending on which security constraint is dominant.

When $\frac{V}{\binom{UV}{G} - \binom{(U-1)V}{G}} 
\ge 
\frac{U - 1}{\binom{UV}{G} - U \binom{V}{G}}$, 
we set $L = L_X = L_Y = V$ and 
$L_S = \big[ \binom{UV}{G} - \binom{(U-1)V}{G} \big] $.

When $\frac{V}{\binom{UV}{G} - \binom{(U-1)V}{G}} 
< 
\frac{U - 1}{\binom{UV}{G} - U \binom{V}{G}}$, 
we set $L = L_X = L_Y = U-1$ and 
$L_S = \big[ \binom{UV}{G} - U \binom{V}{G} \big] $.

We set $W_{u,v}$ as $L$ symbols over the sufficiently large finite field $\mathbb{F}_{q}$, i.e., 
$W_{u,v} \in \mathbb{F}_{q}^{L \times 1}$ and 
$S_{\mathcal{G}} \in \mathbb{F}_{q}^{L_S \times 1}$.

Each user $(u,v)$ transmits
\begin{align}
    X_{u,v} 
    = W_{u,v} 
    + \sum_{\substack{\mathcal{G}: (u,v) \in \mathcal{G}}} 
    \mathbf{H}_{\mathcal{G}}^{u,v} S_\mathcal{G},
    \quad \forall (u,v) \in [U]\times[V].
    \label{Correctness5}
\end{align}

Here, $\mathbf{H}_{\mathcal{G}}^{u,v} \in \mathbb{F}_{q}^{L \times L_S}$ is a precoding matrix satisfying the zero-sum constraint
\begin{align}
    \sum_{(u,v)\in \mathcal{G}} 
    \mathbf{H}_{\mathcal{G}}^{u,v} 
    = 0,
    \quad \forall \mathcal{G} \in \binom{[U]\times[V]}{G}.
    \label{Correctness6}
\end{align}
If $(u,v)\notin \mathcal{G}$, we set $\mathbf{H}_{\mathcal{G}}^{u,v} = 0$.

After receiving messages from its connected users, relay $u$ stacks them as
\begin{align}
\begin{bmatrix}
    X_{u,1} \\
    \vdots \\
    X_{u,V}
\end{bmatrix}
=
\begin{bmatrix}
    W_{u,1} \\
    \vdots \\
    W_{u,V}
\end{bmatrix}
+ \mathbf{H}_1 \mathbf{S}_1,
\quad \forall u \in [U],
\label{Correctness7}
\end{align}

where
 \begin{align}
   &\mathbf{H}_1 \notag\\
   =& \left[ \left\{\mathbf{H}_{\mathcal{G}}^{u,v}\right\}_{\substack{v \in [V], \mathcal{G} \in \binom{[U] \times [V]}{G},\mathcal{G} \notin \binom{[U]\setminus\{u\} \times [V]}{G}}} \right]   \label{Correctness8}\\
  \triangleq& 
  \begin{array}{c@{}c}
   \left[
   \begin{array}{cccc}
     \mathbf{H}_{(1,1),\cdots,(u,1),\cdots,(\left\lceil G/V \right\rceil, \langle G \rangle_V )}^{u,1}, & \cdots      \\
     \vdots\\
  \mathbf{H}_{(1,1),\cdots,(u,1),\cdots,(\left\lceil G/V \right\rceil,\langle G \rangle_V)}^{u,V}, & \cdots  \\
   \end{array}
   \right.
   \end{array}
  \notag\\
   &\begin{array}{c@{}c}
   \left.
   \begin{array}{cccc}
    \cdots & ,\mathbf{H}_{(U-\left\lceil G/V-1 \right\rceil ,V-\langle G \rangle_V+1),\cdots,(u,V),\cdots,(U,V)}^{u,1} \\
    \vdots\\
     \cdots & ,\mathbf{H}_{(U-\left\lceil G/V-1 \right\rceil ,V-\langle G \rangle_V+1),\cdots,(u,V),\cdots,(U,V)}^{u,V}\\
   \end{array}
   \right]
 \end{array}
  \label{matrix_H5} 
 \end{align}

and
\begin{align}
 \mathbf{S}_1 =
     \begin{bmatrix}
     S_{(1,1),\cdots,(u,1),\cdots,(\left\lceil G/V \right\rceil, \langle G \rangle_V)} \\
     \vdots \\
     S_{(U-\left\lceil G/V-1 \right\rceil ,V-\langle G \rangle_V+1),\cdots,(u,V),\cdots,(U,V)}
 \end{bmatrix}. 
 \label{MatrixS1}
\end{align}

Here, $\mathbf{S}_1$ collects all groupwise keys that involve at least one user in relay $u$, i.e.,
$\mathbf{S}_1 = \{ S_{\mathcal{G}} : \mathcal{G} \cap (\{u\}\times[V]) \neq \emptyset \}$.
Accordingly, $\mathbf{H}_1$ is formed by stacking the corresponding $\mathbf{H}_{\mathcal{G}}^{u,v}$ blocks. Therefore, $\mathbf{H}_1$ contains $V$ row blocks and $\binom{UV}{G} - \binom{(U-1)V}{G}$ column blocks.

Each relay then transmits
\begin{align}
    Y_u = \sum_{v\in[V]} X_{u,v}, 
    \quad \forall u \in [U].
    \label{ConstructY}
\end{align}

Stacking all relay transmissions, we obtain
\begin{align}
\begin{bmatrix}
    Y_1 \\
    \vdots \\
    Y_U
\end{bmatrix}
=
\begin{bmatrix}
    \sum_{v\in[V]} W_{1,v} \\
    \vdots \\
    \sum_{v\in[V]} W_{U,v}
\end{bmatrix}
+ \mathbf{H}_2 \mathbf{S}_2,
\label{ServerSecurity0}
\end{align}
where $\mathbf{S}_2$ denotes the effective groupwise keys after relay aggregation, and $\mathbf{H}_2$ is the corresponding precoding matrix.

More specifically, $\mathbf{H}_2$ is constructed by stacking all precoding matrices $\mathbf{H}_{\mathcal{G}}^{u,v}$ associated with groupwise keys, i.e.,
\begin{align}
\mathbf{H}_2 =
\left[ \{\mathbf{H}_{\mathcal{G}}^{u,v}\}_{(u,v)\in[U]\times[V],\ \mathcal{G}\in\binom{[U]\times[V]}{G}} \right],
\label{ServerSecurity10}
\end{align}
\begin{figure*}[ht]         
    \hrulefill             
    \vspace{8pt}           
    \begin{align}
    {\mathbf{H}}_{2}= &
\begin{array}{c@{}c}
  \left[
  \begin{array}{cccc}
      \sum_{\mathcal{G}\in (1,1),\cdots,(\left\lceil G/V \right\rceil,\langle G \rangle_V)}\mathbf{H}_{(1,1),\cdots,(\left\lceil G/V \right\rceil,\langle G \rangle_V)}^{(u,v)\in{\mathcal{M}_1}}, & \cdots      \\
    \vdots\\
\sum_{\mathcal{G}\in (1,1),\cdots,(\left\lceil G/V \right\rceil,\langle G \rangle_V)}\mathbf{H}_{(1,1),\cdots,(\left\lceil G/V \right\rceil,\langle G \rangle_V)}^{(u,v)\in{\mathcal{M}_U}}, & \cdots  \\
  \end{array}
  \right.
  \end{array}
  \notag\\
  &\begin{array}{c@{}c}
  \left.
  \begin{array}{cccc}
   \cdots & ,\sum_{\mathcal{G}\in(U-\left\lceil G/V-1 \right\rceil ,V-\langle G \rangle_V+1),\cdots,(U,V)}\mathbf{H}_{(U-\left\lceil G/V-1 \right\rceil,V-\langle G \rangle_V+1),\cdots,(U,V)}^{(u,v)\in{\mathcal{M}_1}} \\
   \vdots\\
    \cdots & ,\sum_{\mathcal{G}\in (U-\left\lceil G/V-1 \right\rceil ,V-\langle G \rangle_V+1),\cdots,(U,V)} \mathbf{H}_{(U-\left\lceil G/V-1 \right\rceil ,V-\langle G \rangle_V+1),\cdots,(U,V)}^{(u,v)\in{\mathcal{M}_U}}\\
  \end{array}
  \right]
\end{array}
\label{matrix_H7}
\end{align}
\end{figure*}
and $\mathbf{H}_{\mathcal{G}}^{u,v} = 0$ if $(u,v)\notin\mathcal{G}$ is given by \eqref{matrix_H7}, shown at the top of the next page.
The vector $\mathbf{S}_2$ collects all independent groupwise keys $\{S_{\mathcal{G}}\}_{\mathcal{G}\in\binom{[U]\times[V]}{G}}$, i.e.,
\begin{eqnarray}
\mathbf{S}_2=
\begin{bmatrix}
    S_{(1,1),\cdots,(\left\lceil G/V \right\rceil,{\langle G \rangle_V})} \\
    \vdots \\
    S_{(U-\left\lceil G/V-1 \right\rceil ,V-\langle G \rangle_V+1),\cdots,(U,V)}
\end{bmatrix}
\label{MatrixS2_2}
\end{eqnarray}

It is worth noting that when $1 < G \le V$, some groupwise keys are shared only among users within the same relay cluster. In this case, $\mathbf{H}_2$ admits a natural block-wise structure, as illustrated in \eqref{matrix_H7}. Accordingly, $\mathbf{S}_2$ contains all $\binom{UV}{G}$ groupwise keys.

The correctness follows directly from the zero-sum property in \eqref{Correctness6}. Specifically,
\begin{align}
&\sum_{u \in [U]} Y_u\notag\\
=& \sum_{(u,v)\in[U]\times[V]} W_{u,v}
+ \sum_{\mathcal{G}\in\binom{[U]\times[V]}{G}}
\left( \sum_{(u,v)\in\mathcal{G}} \mathbf{H}_{\mathcal{G}}^{u,v} \right) S_{\mathcal{G}} \notag\\
\stackrel{\eqref{Correctness6}}{=}& \sum_{(u,v)\in[U]\times[V]} W_{u,v}.
\label{ServerSecurity5}
\end{align}

This shows that all groupwise key contributions cancel out exactly, ensuring correct aggregation.

To establish security, we next analyze the rank properties of the effective precoding matrices $\mathbf{H}_1$ and $\mathbf{H}_2$, which characterize the linear mixing of the groupwise keys at the relay and server sides, respectively. These rank conditions guarantee that the security constraints at both the relay and the server sides are satisfied.

In particular, the security condition can be reduced to a rank requirement on the corresponding precoding matrices. We formalize this relationship in the following lemmas.

\begin{lemma}\label{lemma:1}
For any relay $u\in[U]$, the scheme \eqref{Correctness5} satisfies the security constraint \eqref{RelaySecurity} if and only if $\mathrm{rank}(\mathbf{H}_1) = VL$ over $\mathbb{F}_{q}$.
\end{lemma}
\begin{proof}
Consider the relay security constraint \eqref{RelaySecurity}.
\begin{align}
   & I\left( \{W_{u,v}\}_{(u,v)\in[U]\times[V]} ; \{X_{u,v}\}_{v \in [V]}  \right)   \label{Relayserver2}\\
    =& H\left( \{X_{u,v}\}_{v\in [V]} \right)\notag\\
    &-H\left(\{X_{u,v}\}_{v\in [V]} \mid  \{W_{u,v}\}_{(u,v)\in[U]\times[V]} \right)\\
\stackrel{\eqref{Correctness7}}{=}&  H\left( \{X_{u,v}\}_{v\in [V]} \right) -H\bigl(\mathbf{H}_1\mathbf{S}_1 \mid    \{W_{u,v}\}_{(u,v)\in[U]\times[V]} \bigr) \label{Relayserver3}\\
\stackrel{\eqref{gruopwisesize}}{=}& H\left( \{X_{u,v}\}_{v\in [V]} \right) -H\left(\mathbf{H}_1\mathbf{S}_1 \right)
    \label{Relayserver4}\\
    =& VL - VL=0
    \label{Relayserver5}
\end{align}
The second term in \eqref{Relayserver3} follows from the construction of $\{X_{u,v}\}_{v\in[V]}$ as described in \eqref{Correctness7}. Step \eqref{Relayserver4} follows from the independence between the inputs and the groupwise keys in \eqref{gruopwisesize}. In \eqref{Relayserver5}, when $\{X_{u,v}\}_{v\in[V]}$ is uniformly distributed, the entropy attains its maximum value (measured in $q$-ary units). In the final step, we assume that the rank of matrix $\mathbf{H}_1$ over $\mathbb{F}_{q}$ is $VL$ and $S_\mathcal{G}$ symbols are i.i.d. and uniform.

For any scheme of the form \eqref{Correctness5}, if \eqref{Relayserver2} equals zero, then all the above inequalities must be tight, which implies that the matrix $\mathbf{H}_1$ must have rank $VL$.
\end{proof}

\begin{lemma}\label{lemma:2}
For any user $(u,v)\in[U]\times[V]$, the scheme \eqref{ConstructY} satisfies the security constraint \eqref{ServerSecurity} if and only if $\mathrm{rank}(\mathbf{H}_2) = (U-1)L$ over $\mathbb{F}_{q}$.
\end{lemma}
\begin{proof}
Consider the server security constraint \eqref{ServerSecurity}.
\begin{align}
    & I\bigg(\{W_{u,v}\}_{(u,v)\in[U]\times[V]}; \{Y_u\}_{u \in [U]} \bigg| \sum_{(u,v)\in[U]\times [V]}{{W}_{u,v}}\bigg)   \notag\\
    =& H\bigg( \{Y_{u}\}_{u \in [U]} \bigg|  \sum_{(u,v)\in[U]\times [V]}{{W}_{u,v}} \bigg) \notag\\
    & - H\bigg( \{Y_{u}\}_{u \in [U]} \bigg| \{W_{u,v}\}_{(u,v) \in [U] \times [V] } \bigg)\label{Rates4}\\
     \stackrel{\eqref{ConstructY}}{=}& H\bigg( \{Y_{u}\}_{u \in [U]} ,  \sum_{(u,v)\in[U]\times [V]}{{W}_{u,v}} \bigg)\notag\\
     &-H\bigg( \sum_{(u,v)\in[U]\times [V]}{{W}_{u,v}} \bigg) \notag\\
    &- H\bigg( \bigg\{\sum_{v\in[V]}X_{u,v}\bigg\}_{u \in [U]} \bigg|\{W_{u,v}\}_{(u,v) \in [U] \times [V] } \bigg)\label{Rates5}\\
    \stackrel{\eqref{Correctness5}}{=}&  H\bigg( \{Y_{u}\}_{u \in [U]}\bigg)-H\bigg( \sum_{(u,v)\in[U]\times [V]}{{W}_{u,v}} \bigg) \notag\\
    &- H\biggm( \bigg\{\sum_{v\in[V]}(W_{u,v}+Z_{u,v})\bigg\}_{u \in [U]} \bigg|  \notag \\
    &  \{W_{u,v}\}_{(u,v) \in [U] \times [V] } \biggm)\label{Rates5-1}\\
   {=}&  H\bigg( \{Y_{u}\}_{u \in [U]}\bigg)-H\bigg( \sum_{(u,v)\in[U]\times [V]}{{W}_{u,v}} \bigg)  - \notag\\
    & H\bigg( \mathbf{H}_2\mathbf{S}_2 \mid \{W_{u,v}\}_{(u,v) \in [U] \times [V]} \bigg)  
    \label{Rates6}\\
    \stackrel{\eqref{gruopwisesize}}{\leq}& (U - 1)L - H\left(\mathbf{H}_2\mathbf{S}_2\right)  \label{Rates7}\\
    =& (U - 1)L - (U - 1)L = 0  \label{Rates8} 
\end{align}
The third term in \eqref{Rates5} follows from the construction of $Y_u$. The first term in \eqref{Rates5-1} follows from the fact that $\sum_{u\in[U]} Y_u = \sum_{(u,v)\in[U]\times[V]} W_{u,v}$. The third term in \eqref{Rates5-1} follows from \eqref{Correctness5}. For the first term in \eqref{Rates7}, the uniform distribution maximizes entropy (measured in $q$-ary units).
The second term of (\ref{Rates7}) follows from the independence between the groupwise keys and the inputs. In the final step, we assume that $\mathrm{rank}(\mathbf{H}_2) = (U-1)L$ and $S_\mathcal{G}$ symbols are i.i.d. and uniform.

For any scheme of the form \eqref{ConstructY}, if \eqref{Rates4} equals zero, then all the above inequalities must be tight, implying that $\mathbf{H}_2$ must have rank $(U-1)L$ over $\mathbb{F}_q$.
\end{proof}

\subsection{\texorpdfstring{Existence of $\mathbf{H}_1$ and $\mathbf{H}_2$}{}}

In this section, we prove that there exist matrices $\mathbf{H}_1$ of the form \eqref{matrix_H5} and $\mathbf{H}_2$ of the form \eqref{matrix_H7} with the desired rank properties.

We relate the existence of $\mathbf{H}_1$ and $\mathbf{H}_2$ to the Schwartz--Zippel lemma. For each relay $u \in [U]$, the precoding matrices $\mathbf{H}^{u,v}_{\mathcal{G}}$ for $v \in [V]$ are generated independently and uniformly over a sufficiently large finite field $\mathbb{F}_q$ as described in \eqref{matrix_H5}. For each $\mathcal{G} \in \binom{[U]\times[V]}{G}$, we arbitrarily select $G-1$ matrices $\mathbf{H}^{u,v}_{\mathcal{G}}$ with $(u,v) \in \mathcal{G}$ and set the remaining one as the negative sum of the others to satisfy \eqref{Correctness6}. This ensures the zero-sum constraint while preserving randomness.

Consider the determinant of a maximal nonsingular submatrix of $\widehat{\mathbf{H}}_1$. This determinant is a multivariate polynomial in the entries of the randomly generated matrices. Since each entry is chosen independently from $\mathbb{F}_q$, the Schwartz--Zippel lemma guarantees that if this polynomial is not identically zero, then it evaluates to a nonzero value with high probability. Specifically, for a nonzero polynomial of total degree $d$, the probability that it evaluates to zero is at most $d/q$. By choosing $q$ sufficiently large, this probability can be made arbitrarily small. Hence, there exists a choice of matrices such that $\operatorname{rank}(\widehat{\mathbf{H}}_1) = VL$, and the probability that this rank condition holds approaches $1$ as $q \to \infty$.

The same argument applies to $\mathbf{H}_2$, where the rank of $\widehat{\mathbf{H}}_2$ is $(U-1)L$. Therefore, the Schwartz--Zippel lemma ensures the existence of $\mathbf{H}_1$ and $\mathbf{H}_2$ with the desired ranks, completing the proof.

\section{Converse}

In this section, we establish information-theoretic lower bounds on the communication rates $R_X$, $R_Y$, and the groupwise key rate $R_S$. 
These bounds characterize the fundamental limits imposed by the correctness and security constraints. 
Since they match the achievable rates derived in Section~\ref{sec:GENERAL SCHEME}, the optimality of the proposed scheme follows.

We divide the proof into two parts. 
First, we show that the problem is infeasible when $G = 1$. 
Then, we consider the feasible regime $1 < G \leq UV$, where we derive the necessary lower bounds on $R_X$, $R_Y$, and $R_S$.

\subsection{\texorpdfstring{Infeasible Regime: $G = 1$}{}}

We partition the set of relays into two nonempty and disjoint subsets $\mathcal{U}_1, \mathcal{U}_2 \subseteq [U]$ such that
$(\mathcal{U}_1 \times [V]) \cup (\mathcal{U}_2 \times [V]) = [U] \times [V],$
with $|\mathcal{U}_1| \ge 1$ and $|\mathcal{U}_2| \ge 1$.
When $G = 1$, each groupwise key is held by a single user and is not shared across users. 
Since the sets $\mathcal{U}_1 \times [V]$ and $\mathcal{U}_2 \times [V]$ are disjoint, it follows that
\begin{align}
    I\left( \{Z_{u,v}\}_{(u,v)\in \mathcal{U}_1 \times [V]} ; \{Z_{u,v}\}_{(u,v)\in \mathcal{U}_2 \times [V]} \right) = 0. \label{independengG}
\end{align}

We next derive a lower bound on the same mutual information term in \eqref{independengG}, which will lead to a contradiction.
\begin{align}
& I\bigl( \{Z_{u,v}\}_{(u,v)\in \mathcal{U}_1 \times [V]}; \{Z_{u,v}\}_{(u,v)\in \mathcal{U}_2 \times [V]}  \bigr) \label{A-87}\\
\stackrel{(\ref{independent2})}{=} & I\bigl( \{W_{u,v}, Z_{u,v}\}_{(u,v)\in \mathcal{U}_1 \times [V]}; \{W_{u,v}, Z_{u,v}\}_{(u,v)\in \mathcal{U}_2 \times [V]}  \bigr)  \label{A-88}\\
\stackrel{(\ref{messageX})}{\ge}& I\bigl( \{W_{u,v}, X_{u,v}\}_{(u,v)\in \mathcal{U}_1 \times [V]}; \{W_{u,v}, X_{u,v}\}_{(u,v)\in \mathcal{U}_2 \times [V]}  \bigr)  \\
\stackrel{(\ref{messageY})}{\ge}& I\bigl( \{W_{u,v}\}_{(u,v)\in \mathcal{U}_1 \times [V]}, \{Y_{u}\}_{u\in \mathcal{U}_1 }; \{W_{u,v}\}_{(u,v)\in \mathcal{U}_2 \times [V]}, \notag \\
& \{Y_{u}\}_{u\in \mathcal{U}_2 }  \bigr)  \\
\ge & I\!\bigg( \sum_{(u,v)\in \mathcal{U}_1 \times [V]} W_{u,v};\, \{W_{u,v}\}_{(u,v)\in \mathcal{U}_2 \times [V]}, \{Y_{u}\}_{u\in \mathcal{U}_2 } \Big| \notag\\
&\{Y_{u}\}_{u\in \mathcal{U}_1} \bigg)  \\
= &H\!\bigg( \sum_{(u,v)\in \mathcal{U}_1 \times [V]} W_{u,v} \Big| \{Y_{u}\}_{u\in \mathcal{U}_1} \bigg)- \notag\\
\quad & \underbrace{H\!\bigg( \sum_{(u,v)\in \mathcal{U}_1 \times [V]} W_{u,v} \Big| \{Y_{u}\}_{u\in \mathcal{U}_1}, W_{\mathcal{U}_2 \times [V]}, \{Y_{u}\}_{u\in \mathcal{U}_2 } \bigg)}_{\stackrel{(\ref{Correctness})}{=}0} \label{A-92} \\
\geq &H\bigg( \sum_{(u,v)\in \mathcal{U}_1 \times [V]} W_{u,v} \Big| \sum_{(u,v)\in [U] \times [V]} W_{u,v},\{Y_{u}\}_{u\in \mathcal{U}_1} \bigg) \\
=& H\!\bigg( \sum_{(u,v)\in \mathcal{U}_1 \times [V]} W_{u,v}  \bigg)  - \notag\\
& I\!\bigg( \sum_{(u,v)\in \mathcal{U}_1 \times [V]} W_{u,v};\, \sum_{(u,v)\in [U] \times [V]} W_{u,v},\{Y_{u}\}_{u\in\mathcal{U}_1}  \bigg)  \label{A-93}\\
   =& H\!\bigg( \sum_{(u,v)\in \mathcal{U}_1 \times [V]} W_{u,v} \bigg) 
   \notag\\
    &- \underbrace{I\!\bigg( \sum_{(u,v)\in \mathcal{U}_1 \times [V]} W_{u,v};\, \sum_{(u,v)\in [U] \times [V]} W_{u,v}  \bigg)}_{\stackrel{(\ref{independent2})}{=}0}  \notag\\
   & - \underbrace{I\!\bigg( \sum_{(u,v)\in \mathcal{U}_1 \times [V]} W_{u,v};\, \{Y_{u}\}_{u\in\mathcal{U}_1} \Big| \sum_{(u,v)\in [U] \times [V]} W_{u,v} \bigg)}_{\stackrel{(\ref{ServerSecurity})}{=}0}   \label{A-944}\\
\geq & L. \label{A-90}
\end{align}
In (\ref{A-88}), we used the independence between the inputs and the keys, together with the fact that $\mathcal{U}_1 \times [V]$ and $\mathcal{U}_2 \times [V]$ are disjoint. 
The second term in \eqref{A-92} is zero due to the correctness constraint in \eqref{Correctness}. 
The second and third terms in \eqref{A-944} are zero due to the independence of the inputs and the server security constraint in \eqref{ServerSecurity}, respectively.

Comparing \eqref{independengG} and \eqref{A-90}, we obtain a contradiction, i.e., $0 \geq L$. 
Therefore, the problem is infeasible when $G = 1$.

\subsection{\texorpdfstring{Feasible Regime: $1 < G \leq UV$}{}}

In this regime, we establish fundamental lower bounds on the communication and groupwise key rates. 
We begin by presenting a key lemma, which shows that each transmitted message must contain at least $L$ symbols even when all other users' inputs and individual keys are revealed. 
This result will be repeatedly used to derive the converse bounds on $R_X$ and $R_Y$. 

\begin{lemma} 
For any $(u,v) \in [U]\times [V]$, we have
\begin{equation}
H\left(X_{u,v} \mid \{W_{i,j}, Z_{i,j}\}_{(i,j) \in ([U]\times [V])\setminus \{(u,v)\}}\right) \geq L.  
\label{eq:step81}
\end{equation}
\end{lemma} 

\begin{proof}
\begin{align}
& H\left(X_{u,v} \mid \{W_{i,j}, Z_{i,j}\}_{(i,j) \in ([U]\times [V])\setminus \{(u,v)\}}\right) \notag\\
\geq& I\Big( X_{u,v}; \sum_{(u',v') \in [U]\times[V]} W_{u',v'} \Big|\notag\\
&\{W_{i,j}, Z_{i,j}\}_{(i,j) \in ([U]\times [V])\setminus \{(u,v)\}} \Big)  \label{eq:step82}\\ 
=& H\Big( \sum_{(u',v') \in [U]\times[V]} W_{u',v'} \Big| \notag \\
&  \{W_{i,j}, Z_{i,j}\}_{(i,j) \in ([U]\times [V])\setminus \{(u,v)\}} \Big) \notag\\
& - H\Big(  \sum_{(u',v') \in [U]\times[V]} W_{u',v'} \Big| X_{u,v},  \\
&  \{W_{i,j}, Z_{i,j}\}_{(i,j) \in ([U]\times [V])\setminus \{(u,v)\}} \Big)   \label{eq:step83}\\
\overset{ \eqref{messageX}}{\geq}& H(W_{u,v})  \notag\\
&- H\Big( \sum_{(u',v') \in [U]\times[V]} W_{u',v'} \Big|  \{X_{u,v}\}_{(u,v) \in [U]\times[V]} \Big)  \label{eq:step84} \\
\overset{\eqref{messageY}}{\geq}& H(W_{u,v}) 
  -\underbrace{H\Big( \sum_{(u',v') \in [U]\times[V]} W_{u',v'} \Big| \{Y_u\}_{u \in [U]} \Big)}_{\overset{(\ref{Correctness})}{=}0}   \label{eq:step85}\\  
=& L.  \label{eq:step86}
\end{align}
\end{proof}
We now explain the key steps in the above derivation. 
The first term in \eqref{eq:step84} follows from the independence of the inputs, since $W_{u,v}$ is independent of $\{W_{i,j}, Z_{i,j}\}_{(i,j)\in ([U]\times [V])\setminus \{(u,v)\}}$. 
The second term in \eqref{eq:step84} follows because $X_{u,v}$ is a deterministic function of $(W_{u,v}, Z_{u,v})$, as given in \eqref{messageX}. 

Next, the inequality in \eqref{eq:step85} follows from the fact that $Y_u$ is a deterministic function of $\{X_{u,v}\}_{v\in[V]}$ (see \eqref{messageY}). 
Finally, the underbraced term in \eqref{eq:step85} is zero due to the correctness constraint in \eqref{Correctness}, which ensures that $\{Y_u\}_{u \in [U]}$ uniquely determines the sum $\sum_{(u',v') \in [U]\times[V]} W_{u',v'}$.

\begin{lemma} 
\label{lem:lemma4}
For any $u \in [U]$, we have
\begin{equation}
H\left(Y_{u} \mid \{W_{i,j}, Z_{i,j}\}_{(i,j) \in ([U]\times [V])\setminus \{(u,v)\}}\right) \geq L.  \label{eq:step87}
\end{equation}
\end{lemma}

\begin{proof}
\begin{align}
& H\left(Y_{u} \mid \{W_{i,j}, Z_{i,j}\}_{(i,j) \in ([U]\times [V])\setminus \{(u,v)\}}\right) \notag\\
\geq& I\Big( Y_{u}; \sum_{(u',v') \in [U]\times[V]} W_{u',v'} \Big| \\
&  \{W_{i,j}, Z_{i,j}\}_{(i,j) \in ([U]\times [V])\setminus \{(u,v)\}} \Big)  \label{eq:step88}\\ 
=& H\Big( \sum_{(u',v') \in [U]\times[V]} W_{u',v'} \Big|  \\
&  \{W_{i,j}, Z_{i,j}\}_{(i,j) \in ([U]\times [V])\setminus \{(u,v)\}} \Big) -\notag\\
&  H\Big(  \sum_{(u',v') \in [U]\times[V]} W_{u',v'} \Big| Y_{u}, \\
& \{W_{i,j}, Z_{i,j}\}_{(i,j) \in ([U]\times [V])\setminus \{(u,v)\}} \Big)   \label{eq:step89}\\
\overset{\eqref{messageX}}{\geq}&  H\Big( \sum_{(u',v') \in [U]\times[V]} W_{u',v'} \Big|  \\
&  \{W_{i,j}, Z_{i,j}\}_{(i,j) \in ([U]\times [V])\setminus \{(u,v)\}} \Big) \notag\\ 
& - H\Big(  \sum_{(u',v') \in [U]\times[V]} W_{u',v'} \Big| Y_{u}, \\
&  \{X_{i,j}\}_{(i,j) \in ([U]\times [V])\setminus \{(u,v)\}} \Big)   \label{eq:step90}\\
\overset{\eqref{messageY}}{\geq}& H\left( W_{u,v}\right) 
  -\underbrace{H\Big( \sum_{(u',v') \in [U]\times[V]} W_{u',v'} \Big| \{Y_u\}_{u \in [U]} \Big)}_{\overset{(\ref{Correctness})}{=}0}  
  \label{eq:step91} \\ 
= & L  \label{eq:step92}
\end{align}
\end{proof}
The second term in \eqref{eq:step90} follows because $X_{i,j}$ is a deterministic function of $(W_{i,j}, Z_{i,j})$ (see \eqref{messageX}). 
The first term in \eqref{eq:step91} follows from the independence of the inputs, since $W_{u,v}$ is independent of $\{W_{i,j}, Z_{i,j}\}_{(i,j)\in [U]\times[V]\setminus\{(u,v)\}}$. 
The second term in \eqref{eq:step91} follows because $Y_u$ is a function of $\{X_{u,v}\}_{v \in [V]}$ (see \eqref{messageY}). 
Finally, the second term in \eqref{eq:step91} is zero due to the correctness constraint in \eqref{Correctness}, which ensures that $\{Y_u\}_{u \in [U]}$ determines $\sum_{(u',v') \in [U]\times[V]} W_{u',v'}$.

\begin{lemma}
\label{lem:lemma5}
For any $u \in [U]$, $\mathcal{V} \subseteq [V]$, we have
\begin{equation}
H\left( \{Z_{u,v}\}_{v \in \mathcal{V}}  \right) \geq |\mathcal{V}|L. \label{eq:step93}
\end{equation}
\end{lemma}
\begin{proof}
\begin{eqnarray}
&&H\left( \{Z_{u,v}\}_{v \in \mathcal{V}}  \right) \notag\\
&\geq& I\left( \{Z_{u,v}\}_{v \in \mathcal{V}} ; \{X_{u,v}\}_{v \in \mathcal{V}} \mid \{W_{u,v}\}_{v \in \mathcal{V}} \right)  \label{eq:step95}\\
&=& H\left( \{X_{u,v}\}_{v \in \mathcal{V}} \mid \{W_{u,v}\}_{v \in \mathcal{V}} \right) \notag\\
  &&- \underbrace{H\left( \{X_{u,v}\}_{v \in \mathcal{V}} \mid \{Z_{u,v}\}_{v \in \mathcal{V}}, \{W_{u,v}\}_{v \in \mathcal{V}} \right)}_{\stackrel{\eqref{messageX}}{=} 0}  \label{eq:step96}\\
&=& H\left( \{X_{u,v}\}_{v \in \mathcal{V}}  \right)  - \underbrace{I\left( \{X_{u,v}\}_{v \in \mathcal{V}} ; \{W_{u,v}\}_{v \in \mathcal{V}}  \right)}_{\stackrel{\eqref{RelaySecurity}}{=} 0}  \label{eq:step97}\\
&=& \sum_{i=1}^{|\mathcal{V}|} H\left( X_{u,v_i} \mid \{X_{u,v_k}\}_{k \in [1,i-1]} \right)  \label{eq:step98}\\
&\geq& \sum_{v \in \mathcal{V}} H\left( X_{u,v} \mid \{X_{u,k}\}_{k \in \mathcal{V} \setminus \{v\}} \right)  \label{eq:step99}\\
&\geq& \sum_{v \in \mathcal{V}} H\left( X_{u,v} \mid \{W_{u,k}, Z_{u,k}\}_{k \in \mathcal{V} \setminus \{v\}}, \{X_{u,k}\}_{k \in \mathcal{V} \setminus \{v\}} \right)  \notag \\
\label{messageX00}\\
&\stackrel{\eqref{messageX}}{=}& \sum_{v \in \mathcal{V}} H\left( X_{u,v} \mid \{W_{u,k}, Z_{u,k}\}_{k \in \mathcal{V} \setminus \{v\}} \right) \label{messageX01}\\
&\stackrel{\eqref{eq:step81}}{\geq}& |\mathcal{V}|L \label{messageX02}
\end{eqnarray}
\end{proof} 
Among them, the second term of \eqref{eq:step97} is zero due to the relay security constraint \eqref{RelaySecurity}. 
In \eqref{eq:step98}, we denote $\mathcal{V} = \{v_1, \dots, v_{|\mathcal{V}|}\}$. 
\eqref{eq:step99} holds because conditioning reduces entropy.

By the above lemmas, we derive the converse bounds on the communication rates $R_X$, $R_Y$, and the groupwise key rate $R_S$.

\textit{Proof of $R_X \geq 1$:} For any $(u,v) \in [U]\times [V]$, we have
\begin{align}
L_X &= H(X_{u,v}) \notag \\
&\geq H\big( X_{u,v} \big| \{W_{i,j}, Z_{i,j}\}_{(i,j) \in [U] \times [V] \setminus \{(u,v)\}} \big)  \stackrel{\eqref{eq:step81}}{\geq} L \label{messageX04}
\end{align}
which implies that $R_X = \frac{L_X}{L} \geq 1$.

\textit{Proof of $R_Y \geq 1$:} For any $u \in [U]$, we have
\begin{align}
L_Y &= H(Y_u) \notag \\
&\geq H\big( Y_u \big| \{W_{i,j}, Z_{i,j}\}_{(i,j) \in [U] \times [V] \setminus \{(u,v)\}} \big)  \stackrel{\eqref{eq:step87}}{\geq} L \label{messageX05}
\end{align}
which implies that $R_Y = \frac{L_Y}{L} \geq 1$.

Next, we establish the converse bound on the groupwise key rate $R_S$. 
The relay security and server security constraints impose two necessary lower bounds on $R_S$, given by
\begin{align}
    R_S \geq \max\left\{
\dfrac{V}{\dbinom{UV}{G} - \dbinom{(U-1)V}{G}},\;
\dfrac{U - 1}{\dbinom{UV}{G} - U \dbinom{V}{G}}
\right\}
\end{align}
We derive these two bounds in the following.

\emph{Proof of $R_S \geq \frac{V}{\binom{UV }{G} - \binom{(U -1)V}{G}}$:}
Consider any relay $u \in [U]$. 
Applying Lemma~\ref{lem:lemma5} with $\mathcal{V} = [V]$, we have
\begin{align}
&VL \notag \\
\stackrel{\eqref{eq:step93}}{\leq}& 
H\left( \{Z_{u,v}\}_{v \in [V]}  \right) \label{messageX07}\\
\overset{(\ref{individualkey})}{=}& H\left( 
    \{S_\mathcal{G}\}_{\substack{\mathcal{G}\in \binom{[U] \times [V]}{G}, 
    \mathcal{G} \cap \mathcal{M}_u \neq \emptyset}}  
\right) \label{messageX08}\\
=& H\left( 
    \{S_\mathcal{G}\}_{\substack{\mathcal{G} \in \binom{[U]\times [V]}{G}}} 
\right) -H\left( 
    \{S_\mathcal{G}\}_{\substack{\mathcal{G} \in \binom{[U]\times [V]}{G}},\mathcal{G} \cap \mathcal{M}_u =\emptyset} 
\right) \notag \\
\label{messageX10}\\
\stackrel{\eqref{gruopwisesize}}{=}& \left[\binom{UV }{G} - \binom{(U -1)V}{G}\right] \times L_S \label{messageX12}
\end{align}
Here, \eqref{messageX08} follows from expressing each $Z_{u,v}$ in terms of the groupwise keys. 
In \eqref{messageX10}, we partition the collection of all subsets $\mathcal{G} \in \binom{[U]\times[V]}{G}$ into those that intersect with $\mathcal{M}_u$ and those that do not. 
Since the groupwise keys $\{S_{\mathcal{G}}\}$ are mutually independent, the total entropy equals the sum of the entropies of the corresponding subsets, which leads to the difference form.

Therefore,
\begin{eqnarray}
R_S  \overset{(\ref{Rates})}{=} \frac{L_S}{L} 
\ge 
\dfrac{V}{\dbinom{UV }{G} - \dbinom{(U -1)V}{G}}. \label{messageX13}
\end{eqnarray}

\emph{Proof of $R_S \geq \frac{U-1}{\binom{UV }{G} - U\binom{V}{G}}$:}
Consider any user $(u,v) \in [U]\times[V]$. We have
\begin{align}
&\quad H\left( \{S_{\mathcal{G}}\}_{\mathcal{G} \in \binom{[U]\times[V]}{G},\; \mathcal{G} \notin \binom{\mathcal{M}_u}{G},\; u \in [U]} \right) \notag \\
& \geq H\left( \{S_{\mathcal{G}}\}_{\mathcal{G} \in \binom{[U]\times[V]}{G},\; \mathcal{G} \notin \binom{\mathcal{M}_u}{G},\; u \in [U]} \mid W_{[U]\times[V]} \right) \label{messageX15} \\
& \geq I\left( \{S_{\mathcal{G}}\}_{\mathcal{G} \in \binom{[U]\times[V]}{G},\; \mathcal{G} \notin \binom{\mathcal{M}_u}{G},\; u \in [U]}; \{Y_u\}_{u \in [U]} \mid W_{[U]\times[V]} \right) \label{messageX16} \\
& = H\left( \{Y_u\}_{u \in[U]} \mid W_{[U]\times[V]} \right) -\notag\\
&\quad  \underbrace{H\big( \{Y_u\}_{u \in[U]} \mid \{S_{\mathcal{G}}\}_{\mathcal{G} \in \binom{[U]\times[V]}{G},\; \mathcal{G} \notin \binom{\mathcal{M}_u}{G},\; u \in [U]}, W_{[U]\times[V]} \big)}_{=0} \label{messageX17} \\
& = H\left( \{Y_u\}_{u \in [U]} \right) - I\left( \{Y_u\}_{u \in [U]}; W_{[U]\times[V]} \right) \notag \\
& = \sum_{j=1}^U H\left( Y_j \mid \{Y_u\}_{u \in [1:j-1]} \right) - \notag\\
&\quad I\Big( \{Y_u\}_{u \in [U]}; W_{[U]\times[V]}, \sum_{(u,v)\in [U]\times [V]} W_{u,v} \Big) \label{messageX19} \\
& \geq \sum_{j=1}^U H\left( Y_j \mid \{Y_u\}_{u \in [U] \setminus \{j\}} \right) - \notag \\
&\quad I\Big( \{Y_u\}_{u \in [U]}; W_{[U]\times[V]}, \sum_{(u,v)\in [U]\times [V]} W_{u,v} \Big) \label{messageX20} \\
& \geq \sum_{j=1}^U H\Big( Y_j \mid \{Y_u\}_{u \in [U] \setminus \{j\}},  \notag \\ 
&\quad  \{W_{u,v}, Z_{u,v}\}_{(u,v) \in ([U] \setminus \{j\}) \times [V]} \Big) \notag \\
&\quad - I\Big( \{Y_u\}_{u \in [U]}; \sum_{(u,v)\in [U]\times [V]} W_{u,v} \Big) \notag\\
&\quad- \underbrace{I\Big( \{Y_u\}_{u \in [U]}; W_{[U]\times[V]} \bigg| \sum_{(u,v)\in [U]\times [V]} W_{u,v} \Big)}_{\stackrel{\eqref{ServerSecurity}}{=} 0} \label{messageX21} \\
& \stackrel{\eqref{messageX}\eqref{messageY}}{\geq} \sum_{j=1}^U H\left( Y_j \mid \{W_{u,v}, Z_{u,v}\}_{(u,v) \in ([U] \times [V]) \setminus \{(j,v_j)\}} \right) \notag \\
& \qquad - H\Big( \sum_{(u,v)\in [U]\times [V]} W_{u,v} \Big) \notag\\
& \qquad + \underbrace{H\Big( \sum_{(u,v)\in [U]\times [V]} W_{u,v} \bigg| \{Y_u\}_{u \in [U]} \Big)}_{\stackrel{\eqref{Correctness}}{=} 0} \label{messageX22} \\
& \stackrel{\eqref{eq:step87}}{\geq} UL - H\Big( \sum_{(u,v) \in ([U] \times [V])} W_{u,v} \Big) \label{messageX23} \\
& = (U-1)L. \label{messageX24}
\end{align}
We now justify the key steps. 
The second term in \eqref{messageX17} is zero because $Y_u$ cannot depend on groupwise keys that are shared exclusively among users within a single relay. 
Otherwise, such keys would introduce components that cannot be canceled across relays, violating the correctness constraint. 
Therefore, $Y_u$ must be a deterministic function of 
$\{W_{u,v}\}_{(u,v)\in \mathcal{M}_u}$ and the groupwise keys that involve users from multiple relays, both of which are included in the conditioning of \eqref{messageX17}. 

The third term in \eqref{messageX21} is zero due to the server security constraint in \eqref{ServerSecurity}. 
The last term in \eqref{messageX22} is zero due to the correctness constraint, which ensures that the sum $\sum_{(u,v)} W_{u,v}$ is recoverable from $\{Y_u\}_{u\in[U]}$. 
Step \eqref{messageX23} follows from Lemma~\ref{lem:lemma4}.

Combining \eqref{messageX15}--\eqref{messageX24}, we obtain
\begin{eqnarray}
  (U-1)L &\leq&  H\left( \{S_{\mathcal{G}}\}_{\mathcal{G} \in \binom{[U]\times[V]}{G},\mathcal{G} \notin \binom{\mathcal{M}_u}{G},u \in [U]}  \right)\label{messageX30}\\
 &{=}&  H\left( \{S_{\mathcal{G}}\}_{\mathcal{G} \in \binom{[U]\times[V]}{G} }  \right)  \notag \\ 
 &&- H\left( \{S_{\mathcal{G}}\}_{\mathcal{G} \in \binom{\mathcal{M}_u}{G},u\in[U] }  \right) \label{messageX31}\\
 &\stackrel{\eqref{gruopwisesize}}{=}&  \left[\binom{UV}{G} -     U\binom{V}{G}\right] L_S. \label{messageX32}
\end{eqnarray}

Therefore,
\begin{align}
R_{S} = \frac{L_S}{L} 
\geq \dfrac{U-1} {\dbinom{UV}{G} - U\dbinom{V}{G}}. \label{messageX33}
\end{align}

In \eqref{messageX31}, we partition the collection of all subsets $\mathcal{G} \in \binom{[U]\times[V]}{G}$ into those that lie entirely within a single relay and those that span multiple relays. 
Since the corresponding groupwise keys are mutually independent, the entropy equals the difference between the total entropy and the entropy of the intra-relay keys.

Finally, the converse bound on the groupwise key rate $R_S$ imposed by relay security follows from \eqref{messageX13}, while the bound imposed by server security follows from \eqref{messageX33}. 
Combining these bounds, we obtain
\begin{align}
    R_S \geq \max\left\{
\dfrac{V}{\dbinom{UV}{G} - \dbinom{(U-1)V}{G}},\;
\dfrac{U - 1}{\dbinom{UV}{G} - U \dbinom{V}{G}}
\right\}.
\end{align}
This completes the proof.

\section{Conclusion}

In this paper, we investigated the capacity of hierarchical secure aggregation with groupwise keys. 
We characterized the fundamental tradeoff among the user-to-relay communication rate $R_X$, the relay-to-server communication rate $R_Y$, and the groupwise key rate $R_S$ under both relay security and server security constraints.

We first established that the problem is infeasible when $G = 1$, demonstrating that nontrivial key sharing across users is necessary for secure aggregation. 
For the feasible regime $1 < G \leq UV$, we proposed a linear coding scheme based on structured groupwise key mixing. 
The security guarantees were characterized through rank conditions on the effective precoding matrices, and the existence of such constructions was ensured via the Schwartz--Zippel lemma.

On the converse side, we derived tight information-theoretic lower bounds on $R_X$, $R_Y$, and $R_S$. 
In particular, we showed that $R_X \geq 1$ and $R_Y \geq 1$, and established matching lower bounds on the groupwise key rate imposed by relay and server security constraints. 
These bounds match the achievable scheme, thereby fully characterizing the capacity region.

Our results highlight the critical role of structured key sharing and linear mixing in achieving optimal secure aggregation. 
Several interesting directions remain open. 
First, it would be of interest to extend the current model to a $T$-secure setting, where up to $T$ relays or users may collude, and to characterize the corresponding capacity region. 
Second, the impact of user dropouts or stragglers on the optimal rates remains unclear and warrants further investigation. 
Third, it is desirable to develop explicit deterministic constructions that achieve the required rank conditions without relying on large field sizes. 
Finally, extending the framework to more general network topologies and heterogeneous connectivity patterns is another promising direction.

\bibliographystyle{IEEEtran}
\bibliography{references_secagg.bib}
\end{document}